\documentclass[11pt]{article}
\usepackage{amsmath,amsthm,amssymb,color}
\newcommand{\remove}[1]{}
\newtheorem{thm}{Theorem}[section]
\newtheorem{claim}[thm]{Claim}
\newtheorem{lem}[thm]{Lemma}
\newtheorem{define}[thm]{Definition}

\newtheorem{conjecture}{Conjecture}
\newtheorem{THM}{Theorem}
\newtheorem{COR}[THM]{Corollary}

\newtheorem{fact}[thm]{Fact}

\newcommand{\comment}[1]{}
\renewcommand{\remove}[1]{}

\newcommand{\eps}{{\varepsilon}}

\newcommand{\cp}{\textnormal{cp}}
\newcommand{\comments}[1]{}
\newcommand{\rank}{\textnormal{rank}}
\newcommand{\colrank}{\textnormal{colrank}}
\newcommand{\spana}{\textnormal{span}}
\newcommand{\Spec}{\textnormal{Spec}}
\newcommand{\colspan}{\textnormal{colspan}}

\def\F{{\mathbb{F}}}
\newcommand{\Z}{\mathbb{Z}}
\newcommand{\R}{\mathbb{R}}
\newcommand{\N}{\mathbb{N}}
\newcommand{\E}{\mathbb{E}}
\newcommand{\C}{\mathcal{C}}

\newcommand{\U}{\mathcal{U}}
\newcommand{\MV}{\mathbf{MV}}
\newcommand{\al}{\alpha}
\renewcommand{\Pr}{\mathbf{Pr}}

\newcommand{\ip}[2]{\langle #1,#2 \rangle}

\def\draft{1}   

\ifnum\draft=1 
    \def\ShowAuthNotes{1}
\else
    \def\ShowAuthNotes{0}
\fi

\ifnum\ShowAuthNotes=0
\newcommand{\authnote}[2]{{ \footnotesize \bf{\color{red}[#1's Note: {\color{blue}#2}]}}}
\else
\newcommand{\authnote}[2]{}
\fi

\newcommand{\bnote}[1]{{\authnote{Abhishek} {#1}}}
\newcommand{\znote}[1]{{\authnote{Zeev} {#1}}}
\newcommand{\snote}[1]{{\authnote{Shachar} {#1}}}

\begin{document}
\title{New Bounds for Matching Vector Families}

\author{
Abhishek Bhowmick\thanks{Department of Computer Science. University of Texas at Austin. Email: \texttt{bhowmick@cs.utexas.edu}. Research done while the author was a visiting student at Princeton University. Research supported by NSF grant CCF-0916160.}
\and
Zeev Dvir\thanks{Department of Computer Science and Department of Mathematics, Princeton University, Princeton NJ.
Email: \texttt{zeev.dvir@gmail.com}. Research partially
supported by NSF grant CCF-0832797 and by the Packard fellowship.}
\and
Shachar Lovett
\thanks{School of Mathematics, Institute for Advanced Study, Princeton, NJ. Email:
\texttt{slovett@math.ias.edu}. Research supported by NSF grant DMS-0835373.}\\
}

\date{} \maketitle
\thispagestyle{empty}
\begin{abstract}


A Matching Vector (MV) family modulo $m$ is a pair of ordered lists $U=(u_1,\ldots,u_t)$ and $V=(v_1,\ldots,v_t)$ 
where $u_i,v_j \in \Z_m^n$ with the following inner product pattern: for any $i$, $\langle u_i,v_i\rangle=0$, and 
for any $i \ne j$, $\langle u_i,v_j\rangle \ne 0$. A MV family is called $q$-restricted if inner products $\langle u_i,v_j\rangle$ take 
at most $q$ different values. 

Our interest in MV families stems from their recent application in the construction of sub-exponential locally 
decodable codes (LDCs). There, $q$-restricted MV families are used to construct LDCs with $q$ queries, and there is special
interest in the regime where $q$ is constant. When $m$ is a prime it is known that such constructions yield codes
with exponential block length. However, for composite $m$ the behaviour is dramatically different. 
A recent work by Efremenko \cite{Efremenko} (based on an approach initiated by Yekhanin \cite{Y_nice}) gives the first sub-exponential
LDC with constant queries. It is based on a construction of a MV family of super-polynomial size by Grolmusz \cite{Grolmusz}
modulo composite $m$.

In this work, we prove two lower bounds on the block length of LDCs which are based on black box construction using MV families. When $q$ is constant 
(or sufficiently small), we prove that such LDCs must have a quadratic block length. When the modulus $m$ is constant 
(as it is in the construction of Efremenko \cite{Efremenko}) we prove a super-polynomial lower bound on the block-length of the LDCs,
assuming a well-known conjecture in additive combinatorics, the polynomial Freiman-Ruzsa conjecture over $\Z_m$.

\end{abstract}

\newpage

\section{Introduction}

A Matching Vector Family (MV Family) is a combinatorial object that arises in several contexts including Ramsey graphs, weak representation of OR polynomials  and recently in constant query locally decodable codes (LDCs). It is defined by two ordered lists $U=\left(u_1, \cdots u_t\right)$ and $V=\left(v_1, \cdots v_t\right)$ where $u_i, v_j \in \Z_m^n$ and $m$ and $n$ are integers greater than $1$. The property that the two lists have to satisfy is the following: for all $i\in [t]$, $\langle u_i, v_i \rangle = 0\ (mod\ m)$  whereas for all $i \neq j\in [t]$, $\langle u_i, v_j \rangle \neq 0\ (mod\ m)$. By $\langle \cdot , \cdot \rangle$ we denote the standard inner product. Let us call this the standard definition of a MV family. If in addition, all the inner products $\langle u_i, v_j \rangle\ (mod\ m)$ lie in a set of size $q$, then it is called a $q-restricted$ MV family. Note that $q=m$ corresponds to the standard MV family. The size of the MV family is $t$, the number of vectors in the list.
In this paper, we shall prove upper bounds on $q-restricted$ MV families in the first part and on standard MV families in the later part.

Let $\MV(m,n)$ denote the largest $t$ such that there exists a MV family of size $t$ in $\Z_m^n$. Analogously, let $\MV(m,n,q)$ denote the largest $t$ such that there exists a $q-restricted$ MV family of size $t$ in $\Z_m^n$. The question of bounding $\MV(m,n)$ (or $\MV(m,n,q)$) is closely related to the well-known combinatorial problem of set systems with restricted modular intersections \cite{BF,Sgall,Grolmusz,Grolmusz_2} (in this setting the vectors $u_i,v_i$ are required to have entries that are either $0$ or $1$). The systematic study of this more general problem, in the context of MV codes, was initiated in \cite{DGY}. The setting of prime $m$ is well understood. For large prime $m=p$, it is known that $\MV(p,n)=O\left(p^{n/2}\right)$ \cite{DGY}. Infact, this is almost tight. When $m$ is a small prime, again we have a tight upper bound of  $O\left(n^{p-1}\right)$ \cite{BF}. Surprisingly, the setting of small composite $m$ leads to very useful constructions of Ramsey graphs and constant query LDCs. This is due to a construction of MV family over $Z_6$ by Grolmusz \cite{Grolmusz} of superpolynomial size in contrast to a polynomial upper bound when $m$ is a small prime. Thus, it is interesting to study the behavior of MV families for small composite $m$, and more generally arbitrary general composites. We will see later the connection between upper bounds on $\MV(m,n,q)$ and lower bounds on the encoding lengths of MV Codes (a family of LDCs).  For general $m$, the best upper bound known \cite{DGY} is $\MV(m,n) \leq m^{n-1 + o_m(1)},$ with $o_m(1)$ denoting a function that goes to zero when $m$ grows. It was conjectured in \cite{DGY} that an upper bound of $ \sim m^{n/2}$ should hold for any $m$ (not just prime). This would be tight for large $m$ as there are constructions of MV families almost meeting this bound \cite{YGH12}. However, the proof method used in \cite{DGY} to prove the $O\left(p^{n/2}\right)$ bound does not extend to non primes. In this work, we prove the conjecture for $q-restricted$ MV families in $\Z_m^n$, for any $m$ as long as $q= \frac{o(n)\log m}{\log \left(o(n)\log m\right)}$ (See Theorem \ref{thm-thm1}).
When $m=p$ is a fixed small prime, it follows from \cite{BF} that $\MV(p,n) = O\left(n^{p-1}\right)$. On the other hand, when $m$ is a fixed composite, say $m=6$, there exists a MV family of superpolynomial size $\Omega\left(\exp\left(\log^2n /\log \log n\right)\right)$) \cite{Grolmusz}. We prove a stronger upper bound on $\MV(m,n)$, compared to Theorem \ref{thm-thm1} in such a case assuming a well known conjecture in additive combinatorics (see Theorem \ref{thm-thm2}). Table \ref{table} lists the known and new upper bounds on MV families.

\begin{table}[h]
\centering

\begin{tabular}{|c|c|}
\hline
m&$\MV(m,n)$ or $\MV(m,n,q)$\\ \hline
general prime & $\MV(m,n)\leq O\left(m^{n/2}\right)$ \cite{DGY}\\ \hline
general composite & $\MV(m,n,q) \leq q^{O(q\log q)}m^{n/2}$ (Theorem \ref{thm-thm1})\\ \hline
small, fixed prime & $\MV(m,n) \leq O\left(n^{m-1}\right)$ \cite{BF}\\ \hline
small, fixed composite & $\MV(m,n) \leq 2^{O_m\left(n/\log n\right)}$ (Theorem \ref{thm-thm2}  under Conjecture \ref{pfr}) \\ \hline
\end{tabular}
\label{table}
\caption{List of upper bounds on $\MV(m,n)$, $\MV(m,n,q)$}
\end{table}


\begin{THM}\label{thm-thm1}
For all $m\ge 2, n \ge 1$ we have $$ \MV(m,n,q) \leq q^{O(q\log q)}m^{n/2} $$
\end{THM}
Hence, Theorem~\ref{thm-thm1} resolves the conjecture of \cite{DGY} for any $m$ and for $q = \frac{o(n)\log m}{\log \left(o(n)\log m\right)}$. When $m >> n$, our bound is quite close to the best known construction of MV families which gives  $\MV(m,n) \geq \left( \frac{m+1}{n-2} \right)^{n/2 - 1}$~\cite{YGH12}.

Our second result assumes the polynomial Freiman-Ruzsa conjecture (PFR) conjecture (discussed below) and gives a stronger upper bound on the size of MV families when $m$ is a constant and $n$ grows.

Before we state the conjecture, we need to define what a difference set is. For an abelian group $G$ let $A \subseteq G$. Then the difference set \[A-A=\{a_1-a_2:a_1,a_2 \in A\}\]

\begin{conjecture}[PFR Conjecture in $\Z_m^n$]\label{pfr} Suppose $A \subseteq \Z_m^n$ and $|A-A|\leq \lambda\cdot|A|$. Then one can find a subgroup $H$ of size at most $|A|$ such that $A$ can be covered by $\lambda' = \lambda^{c_m}$ many translates of $H$, where $c_m$ depends only on $m$.
\end{conjecture}

\snote{added discussion on PFR in CS}

We note that the PFR conjecture has already found several applications in computer science. Ben-Sasson and Zewi \cite{SZ} used it to construct two-source extractors from affine extractors; and Ben-Sasson, Lovett and Zewi \cite{SLZ} used it to bound the deterministic communication complexity of functions whose corresponding matrix has low rank. Our work provides another application for the PFR and demonstrates its wide-reaching applicability. We further note that a quasi-polynomial version of the PFR conjecture was recently proved by Sanders \cite{Sanders} (see also the exposition in \cite{Lovett-exposition}). Unfortunately, all the applications discussed above require the truly polynomial version of the conjecture, and so cannot apply to Sanders' result.

We now state the second theorem.
\begin{THM}\label{thm-thm2}
Assuming the PFR conjecture over $\Z_m^n$ (Conjecture~\ref{pfr}) we have $$\MV(m,n) \leq \exp\left( c(m)\frac{n}{\log n} \right), $$ with $c(m)$ an explicit function of $m$.
\end{THM}

\snote{changed next paragraph}

From a technical point of view, one of the ingredients in this work builds on the recent work of Ben-Sasson, Lovett and Zewi \cite{SLZ} who used the PFR conjecture to show that matrices over $\Z_2$ with large bias (say, with many more ones than zeros) and small rank must contain a large monochromatic sub-matrix. An important ingredient in our proof is a generalization of their results from $\Z_2$ to $\Z_m$ for all $m$, not necessarily prime. We note however that this is just one ingredient in our overall proof.

\subsection{Lower Bounds on LDCs: Motivation for MV Family}
Locally Decodable Codes (LDCs) are a special kind of Error Correcting Codes (ECCs) that allow the receiver  to decode a {\em single} symbol of the message by querying a small number of positions in a corrupted encoding. More formally, an $(q,\delta,\epsilon)$-LDC encodes $K$-symbol messages $x$ to $N$-symbol codewords $C(x),$ such that for every $i\in [K],$ the symbol $x_i$ can be recovered with probability $1-\epsilon,$ by a randomized decoding procedure that makes only $q$ queries, even if the codeword $C(x)$ is corrupted in up to $\delta N$ locations. Since the early 90's, LDC's have found exciting applications in various areas ranging from data transmission to  complexity theory to cryptography/privacy. We refer the reader to ~\cite{T_survey,Y_now} for more background.

A central research question, which is far from being solved,  has to do with understanding the best possible `stretch' of an LDC with a constant number of queries. That is, how large $N$ has to be as a function of $K$ for constant $q$ and with constant $\delta,\eps$ (these two last parameters are not our focus here and we will generally assume them to be small fixed constants). For $q=1,2$ this question is completely answered. There are no LDC's for $r=1$ \cite{KT} and the best LDC's with $q=2$ have exponential encoding length \cite{GKST,KdW}. For $q > 2$ there are huge gaps in our understanding. Katz and Trevisan were the first to study this problem \cite{KT} and, today, the best general lower bounds on $N$ are slightly super-linear bounds of the form $\tilde\Omega\left(K^{1+1/\left(\lceil r/2\rceil -1\right)}\right)$~\cite{Woodruff}. Notice that, when the number of queries is 3 or 4, these bounds are quadratic (see also \cite{KdW,Wood10} for the $q=3,4$ case). The upper bounds were, until recently, those coming from polynomial codes and were of the order of $N \leq \exp\left(K^{\frac{1}{q-1}}\right)$. Improved upper bounds, breaking this barrier slightly, were given in \cite{BIKR}.


This state of affairs changed dramatically when, in a breakthrough paper, Yekhanin \cite{Y_nice} developed a new approach for constructing LDCs, called MV codes, that have much shorter codeword length than polynomial codes. Efremenko \cite{Efremenko} was the first to show that this approach could yield codes with subexponential encoding length (Yekhanin's paper showed this under a number theoretic assumption). More refinements and improvements to this new framework were obtained \cite{Ragh,KY,Itoh_Suzuki,MFLWZ,DGY,BET} to give LDC's with $q$ queries and with encoding length that grows, when $q$ is a constant, roughly like $$N \sim \exp\exp\left((\log K)^{O(1/\log q)}(\log\log K)^{1-1/\log q}\right).$$ 

While significantly smaller than the length of polynomial codes, the codeword length of these new codes is still super polynomial in $K$. The most general setting of parameters was addressed in \cite{DGY} where the authors had given a black box construction of $q$ query MV codes using $q-restricted$ MV families in $\Z_m^n$. 
Using the standard definition of MV families, this implied $m$ query MV codes using MV families in $Z_m^n$. In this basic, yet general reduction, it was shown that upper bounds on MV families would lead to lower bounds on the encoding length of MV codes. With this motivation in mind, the authors in \cite{DGY} made a conjecture on the upper bound on the size of MV families which would lead to lower bounds on the encoding length of MV codes under the basic framework. We note that Yekhanin in \cite{Y_nice} used restricted MV families in $\Z_p^n$ where $p$ is a very large Mersenne prime and used a specialized technique to reduce the number of queries from $p$ to $3$. Another instance of reduction in the number of queries from what the standard construction gives, was given by Efremenko \cite{Efremenko} where he again used restricted MV families. A certain gadget was discovered using computer search whereby the author worked in $\Z_{511}$ but got down the number of queries to $3$ from the basic bound of $511$.

The following is a corollary of Theorem \ref{thm-thm1}.
\begin{COR}\label{cor-thm1}
For an arbitrary positive integer $m$, consider an infinite family of $q$-query Matching Vector code $C_n:\F^k \rightarrow \F^N$ for $n \in \N$, where $k(n)$ and $N(n)$ are growing functions of $n$, constructed using the black box reduction from a $q$-restricted Matching Vector Family in $\Z_m^n$ (\cite{DGY}). For large enough $n$, if  $q=\frac{o(n)\log m}{\log \left(o(n)\log m\right)}$, then \[ N \geq k^{2-o(1)}\]
Specifically, if $q=O(1)$, then $N = \Omega\left(k^{2}\right)$.
\end{COR}

Next we have the following corollary from Theorem \ref{thm-thm2}.
\begin{COR}\label{cor-thm2}
For some arbitrary positive integer $m$, assume the PFR conjecture over $\Z_m^n$ (Conjecture~\ref{pfr}). Consider an infinite family of $m$-query Matching Vector code $C_n:\F_q^k \rightarrow \F_q^N$ for $n \in \N$, where $k(n)$ and $N(n)$ are growing functions of $n$, constructed using the black box reduction from a standard Matching Vector Family in $\Z_m^n$ (\cite{DGY}). For large enough $n$, if  $m=O(1)$, then \[ N = \exp\left(\Omega_m\left(\log k \ \log \log k\right)\right)\]
\end{COR}

Thus Corollary \ref{cor-thm2} states that, assuming Conjecture \ref{pfr}, MV codes with constant number of queries must have super polynomial encoding length in the basic framework. Note that we get the same bound in Efremenko's framework for $3$ queries. This is because the form of the superpolynomial bound is assuming a constant $m$ and applying our bound to Efremenko's work again leads to a superpolynomial bound as $m=511$ in his setting (another constant). (He uses $\Z_{511}$ to construct the MV family and further reduces the number of queries to $3$.) This essentially means that in order to construct polynomial length codes, one needs to construct MV families in $\Z_m^n$ for non-constant $m$ and use some specialized gadget to reduce the number of queries. One way is to ensure it is a $q-restricted$ (constant $q$) MV family. This automatically ensures $q$ query decoding. However, the quadratic lower bound continues to hold even in this scenario for constant $q$. To beat the quadratic lower bound for constant query MV codes, one needs to construct $q-restricted$ MV families for growing $m$ and $q=\frac{o(n)\log m}{\log \left(o(n)\log m\right)}$ and then develop some special gadget to get the number of queries down further from $q$ to some constant.

\subsection{Proof Overview}

 The proof of Theorem~\ref{thm-thm1} relies on intuitions coming from the theory of two-source extractors \cite{Chor-Goldreich}, which are functions of two variables $F(X,Y)$ such that the output of $F$ is distributed in a close-to-uniform fashion whenever the two inputs are drawn, independently, from two distributions of sufficiently high entropy. Since our proof does not use two-source extractors explicitly we do not define them formally and just use them to explain the high level idea behind the proof. It is a well known fact \cite{Chor-Goldreich} that the inner product function $F(X,Y) = \ip{X}{Y}$, say over $\Z_2^n \times \Z_2^n$ is a good two source extractor when the two inputs $X$ and $Y$ are both drawn uniformly from sets $S_X,S_Y \subseteq \Z_2^n$ of size larger $2^{n/2}$. This immediately suggests a connection to MV families, since, if we take $S_X = U$ and $S_Y = V$ for a MV family $U,V$ in $\Z_2^n$, we would get a completely non-uniform output (it will be zero with exponentially small probability). This means that the size of $U,V$ is bounded from above by approximately $2^{n/2}$.
 \snote{added approximately}

 If we try   to use a similar  argument over $\Z_m$ we run into trouble since the inner product function modulo $m$ is {\em not} a good two source extractors for sources of size $m^{n/2}$. Take, for example, $S_X = S_Y=\{0,2,4\}^n \subseteq \Z_6^n$ and observe that $\ip{X}{Y}$ is always divisible by 2 and so is far from being uniformly distributed over $\Z_6$. It is, however, possible to show that this  example is, in some sense, the only example and that, in general, we can always find a certain number of elements of either $S_X$ or $S_Y$ that `agree' modulo some factor of $m$. This observation suggests proving Theorem~\ref{thm-thm1} by induction on the number of factors of $m$, which is the way we proceed.

The proof of Theorem~\ref{thm-thm2} uses a slightly different view of MV families as matrices with certain zero/non-zero pattern and small rank. Specifically, for a MV family $U,V$ of size $t$ in $\Z_m^n$ consider the $t \times t$ matrix $P$ whose $(i,j)$'th entry is $\ip{u_i}{v_j} \mod m$. The definition of a MV family implies that $P$ has zeros on the diagonal and non-zeros everywhere else. If $m$ was a prime, we could think of $\Z_m$ as a field $\F$ and say that, since $P$ is the inner product matrix of vectors of length $n$ over a field, it must have rank at most $n$. Conversely, every $t \times t$ matrix over a field $\F$ with these properties (zero on the diagonal and non-zero off the diagonal) and with rank $n$  gives a MV family of size $t$ in $\F^n$. \znote{this part of the discussion is for fields. } We can call  a matrix with this pattern of zeros/non-zeros an {\em MV matrix}. Thus, when $m$ is prime, the question of bounding the size of a MV family is the same as lower bounding the rank of a MV-matrix\footnote{For technical reasons, the actual proof will not be entirely using matrices and will keep the MV family in the background. This is because we need to keep certain invariants throughout the proof and these are easier to define for families of vectors than for matrices.}. When $m$ is composite, this whole approach should be re-examined since $\Z_m$ is no longer a field and our familiar understanding of matrices and linear algebra over a field are no longer valid.  We do, however, manage to carry over this correspondence between the two problems by defining the notion of rank in a careful way (more on this issue below).

Assume for the purpose of this overview that the usual notion of rank and other intuitions from linear algebra are  valid  over $\Z_m$ and let us proceed with sketching the proof of Theorem~\ref{thm-thm2} using the equivalent formulation as bounding (from below) the rank of a MV matrix $P$. The starting point  is a generalization of a result of \cite{SLZ}, mentioned above, from $\Z_2$ to $\Z_m$. We show that every matrix $P$ over $\Z_m$ that is biased (i.e., its values are not distributed close to uniformly) and has low rank, contains a large monochromatic sub-matrix {\em modulo some factor $m'$ of $m$}. The size of the sub-matrix is bounded from below by $\sim |P|\exp(-r'/\log(r'))$, where $r'$ is the rank of $P$ modulo $m'$ (this factor depends on the specific way the matrix is biased). This generalizes the result of \cite{SLZ} which assumes $m=2$ and finds a large monochromatic sub-matrix (modulo 2). We note that the sub-matrix lemma is the only component in the proof that relies on the PFR conjecture. Let us refer to this result from now on as the {\em sub-matrix} lemma. We can apply the sub-matrix lemma to a MV matrix $P$ since its values are far from uniform (the probability of zero is much less than $1/m$) and since its rank is assumed (towards a contradiction) to be low.

Suppose for the sake of simplicity that $m=p \cdot q$, with $p,q$ distinct primes (the proof for general $m$ is significantly more technical but relies on the same basic intuitions). Applying the sub-matrix lemma we obtain a   sub-matrix $P_1$ of $P$ that is constant modulo some factor $m_1$ of $m$ (so $m_1$ is either $p$, $q$ or $m$) of size at least $|P|\exp(-r_1/\log(r_1))$, where $r_1\leq n$ is the rank of $P \mod m_1$. Using some matrix manipulations, and subtracting a rank one matrix, we can get a large sub-matrix $P_1'$ that does not intersect the diagonal of $P$ and s.t all of the entries of $P_1'$ are zero modulo $m_1$. Suppose $|P_1'| = t_1$ and consider the $2t_1 \times 2t_1$ sub-matrix $P_1''$ of $P$ that has $P_1'$ as its top-right (or bottom-left) block and s.t the top-left and bottom-right blocks are taken to have zero diagonal elements. Formally, if $P_1'$ is indexed by rows in $R$ and columns in $T$ with $R \cap T = \emptyset$ then the rows/columns of $P_1''$ will be indexed by $R \cup T$. If we consider the matrix $P_1''$ modulo $m_1$ then it has top-right block which is all zero and so its rank (modulo $m_1$) will be the sum of the ranks of the top-left and bottom right blocks. Thus, one of these blocks, w.l.o.g the top-left one, must have rank at most $n/2$ (over $\Z_{m_1}$). Notice also that both of these blocks are themselves MV matrices modulo $m$ since they are sub-matrices of $P$ with the same row and column sets. Let $\tilde P_1$ be the top-left block of $P_1''$. We can now apply, again, the monochromatic sub-matrix lemma to find a large  sub-matrix $P_2$ of $\tilde P_1$ which is constant modulo some other factor $m_2$ of $m$. The size of $P_2$ will be $$t_1 \cdot \exp(-r_2/\log(r_2)) = |P|\cdot \exp(-r_1/\log(r_1) - r_2/\log(r_2) ).$$ The factor $m_2$ is also either $p$ or $q$. If it happens to be that $m_1=m_2$ then $r_2 \leq n/2$ and so we gain in the size of $P_2$ in this second step (the expression $r_2/\log(r_2)$ is smaller than $n/2\log(n/2)$ which is smaller by roughly a factor of two than our bound on  $r_1/\log(r_1)$. Suppose we continue with this iterative process of finding constant sub-matrices for $\ell$ steps and that, by luck, all the factors $m_1,m_2, \ldots$ are equal to the same factor of $m$ (say $p$). Then, after roughly $\log(n)$ iteration, we will reduce the rank modulo $p$ to one and still have at least $$|P|\cdot \exp\left(-\sum_{i=1}^\ell \frac{n}{2^i\log(n/2^i)}\right)$$ rows, which is close to the original size of $P$ if we assume (in contradiction) that $|P| >> \exp( n/\log n)$ (this calculation is given in Claim~\ref{clm-sum}). In this case we obtain a new large MV family $U',V'$ modulo $m$ such that all inner products $\ip{u_i'}{v_j'}$ of elements $u_i'\in U', v_j' \in V'$ are fixed modulo $p$. From this we can easily construct a MV family of roughly the  same size in $\Z_q^n$ and then use the bounds on $\MV(q,n)$ for primes to get a contradiction. In the `unlucky' case we will have different factors $m_1,m_2,\ldots$ in each stage, but we can adapt the analysis to consider the decrease in rank simultaneously for all factors of $m$.

The full proof is by induction on the number of factors of $m$ and uses the iterative sub-matrix argument above to go from a MV family modulo $m$ to a MV family of roughly the same size modulo some proper factor of $m$ (and then uses the inductive hypothesis on this new MV family).

\subsection{Matrix rank over $\Z_m$} An important technical issue, which was already hinted at above, is in the definition of the rank of a matrix with entries in a ring $\Z_m$. There are two main properties of matrix rank over a field that we relied on in the  proof sketch above. The first is that a rank $r$ matrix is always the inner product matrix of vectors in $r$ dimensions. Equivalently, a $t \times t$ matrix of rank $r$ can be written as a product of a $t \times r$ matrix and an $r \times t$ matrix. This is important if we are to go back and forth between matrices and MV families. Another property we used is that, if we have a $2t \times 2t$ matrix composed of 4 blocks of size $t \times t$ and the top-right block is zero, then the rank of the matrix is the sum of the ranks of the top-left block and the bottom right block.

Ideally, we would like to define rank over $\Z_m$ so that both properties are satisfied. This is, however, impossible as the following example shows: Consider the $2 \times 2$ matrix with the two rows $(4,0)$ and $(0,3)$ over $\Z_6$. This matrix can be written as the product of the two vectors $(2,3)^T$ and $(2,3)$ and so should have rank one, if we are to satisfy the first property. However, if we are to satisfy the second property, its rank should be the sum of the ranks of the two $1 \times 1$ matrices $(4)$ and $(3)$, which clearly cannot have rank zero!

Our solution to this problem is to give two different definitions of rank, each satisfying one of the two properties. We then show that the two definitions of rank can differ from each other by a multiplicative factor of $\log m$, which our proof can handle. The first definition of rank is as the smallest $r$ such that our  $t \times t$ matrix can be written as a product of a $t \times r$ matrix and an $r \times t$ matrix. Clearly this would satisfy the first property (but not the second). The second definition of rank is termed {\em column-rank} and is defined as the logarithm to the base $m$ of the size of the additive subgroup of $\Z_m^t$ generated by the columns of the matrix. Notice that this definition of rank can result in the rank being non-integer. For example, the rank of the matrix with a single column $(2,0)$ over $\Z_6$ would be equal to $\log_6(3)$ since the subgroup generated by this column is composed of the three vectors $(2,0),(4,0),(0,0)$. It is not hard to show (see Claim~\ref{clm-coltriangular}) that this definition satisfies the second property described above regarding block matrices. Clearly, the two definitions agree for matrices over a field. We show (see Claim~\ref{clm-rankcolrank}) that the two notions of rank can differ by a multiplicative factor of at most $\log m$. This allows us to use both definitions in different parts of the proof without losing too much in the transition. We finish this discussion by noting that in no part of the proof do we use the characterization of rank using determinants, which is often very useful when working over a field.

\subsection{Organization}
We begin with some preliminaries in Section \ref{sec-prelim}. We prove Theorem \ref{thm-thm1} in Section \ref{sec-buildup}. Section \ref{sec-zmmatrices} contains some claims about matrices over $\Z_m$. Section \ref{sec-collfree} introduces collision free MV families. Both Section \ref{sec-zmmatrices} and Section \ref{sec-collfree} will be used in the proof of Theorem \ref{thm-thm2} in Section \ref{sec-main}. The proof of Theorem \ref{thm-thm2} also requires the sub-matrix lemma, whose proof appears in Section \ref{sec-mono}.

\section{General preliminaries}\label{sec-prelim}

\paragraph{Notations:}  Throughout the paper we will be handling ordered lists of elements. A list $A$  of size $t$ over a finite set $\Omega$ is an ordered $t$-tuple $A=\left(a_1,a_2,\cdots ,a_t\right)$ where each $a_i \in \Omega$. A list can have repetitions. If it doesn't we say it is {\em twin free}.  When discussing sublists $A \subseteq B$ with $B = (b_1,\ldots,b_t)$ we will use the convention that, unless specified otherwise, $A$ maintains the ordering induced by $B$. For a positive integer $t$, we let $[t]$ denote the list $\left(1,\cdots t\right)$. So, for example, when we say that $T \subseteq [t]$ we mean that $T$ is a list of integers in increasing order belonging to $[t]$. We say that a list $A = (a_1,\ldots,a_t)$ over $\Omega$ is {\em constant} if $a_i=a_j$ for all $i,j \in [t]$. We assume all logarithms are in base $2$ unless otherwise specified.
\snote{added last sentence so we don't have to write $\log$ all the time}

\subsection{MV Families: Basic Facts and Definitions}
We now start with some basic definition and claims regarding MV families.

\begin{define}[Matching Vector Family]Let $U=\left(u_1, u_2, \cdots u_t\right)$ and $V=\left(v_1, v_2, \cdots v_t\right)$ be lists over $\Z_m^n$. Then $\left(U,V\right)$ is called a \emph{matching vector family} of size $t$ in $\Z_m^n$ if
\begin{itemize}
\item $\langle u_i,v_i \rangle=0\ \left(mod\ m\right), \quad \forall i$.
\item $\langle u_i,v_j \rangle \neq 0\ \left(mod\ m\right), \quad \forall i \neq j$.
\end{itemize}
If in addition, we $|\{\langle u,v \rangle: u \in U, v \in V\}|=q$, we call such a MV family an $q-restricted$ MV family.
We denote the size of $\left(U,V\right)$ by $\left|\left(U,V\right)\right|$. For instance, $\left|\left(U,V\right)\right|=t$ above.
\end{define}

\begin{define}[Subset of Matching Vector Family]Let $U=\left(u_1, u_2, \cdots u_t\right), V=\left(v_1, v_2, \cdots v_t\right)$ form a matching vector family in $\Z_m^n$ of size $t$. By $\left(U',V'\right)\subseteq \left(U,V\right)$, we mean there exists a sublist $T \subseteq [t]$ such that $U'=\left(u_i:i \in T\right), V'=\left(v_i:i \in T\right)$. Observe that $\left(U',V'\right)$ is a matching vector family in $\Z_m^n$.
\end{define}

\begin{define}[$\MV\left(m,n\right)$] We denote by $\MV\left(m,n\right)$ the maximum size of a matching vector family
$\left(U,V\right)$ in $\Z_m^n$. 
Similarly, we denote by $\MV\left(m,n,q\right)$ the maximum size of an $q-restricted$ matching vector family
$\left(U,V\right)$ in $\Z_m^n$. 
\end{define}

We shall use the following simple facts implicitly throughout the paper.
\begin{fact}\label{incr}$\MV\left(m,n\right)$ is an increasing function of $n$.\end{fact}
\snote{shortened proof}
\begin{proof}
	For $n_1<n_2$, we show $\MV\left(m,n_1\right)\leq \MV\left(m,n_2\right)$. Given $\left(U,V\right)$, a matching vector family in $\Z_m^{n_1}$, we can pad each element in $U$ and $V$ by $n_2-n_1$ zeros and obtain a matching vector family in $\Z_m^{n_2}$ of the same size.
\end{proof}

\begin{fact}\label{twinfree}If $\left(U,V\right)$ is a matching vector family in $\Z_m^n$, then $U$ and $V$ are twin free.\end{fact}
\begin{proof}
	Let $U=\left(u_1,u_2,\cdots u_t\right), V=\left(v_1,v_2,\cdots v_t\right)$. We prove $U$ is twin free. By symmetry $V$ is also twin free. Suppose $u_i=u_j$ for some $i \neq j$. Now, $\langle u_i,v_j\rangle=\langle u_j,v_j\rangle=0$ which is a contradiction.
\end{proof}

To facilitate writing in the proofs to follow we introduce the following notation for taking lists, matrices, etc. modulo an integer $r$. \snote{Changed below to just r which divides m}

\begin{define}[Modulo $r$ notation]\label{modulus} Let $2 \le r \le m$ be such that $r$ divides $m$.
Given $a=\left(a_1, \cdots, a_n\right) \in \Z_m^n$, we denote by $a^{\left(r\right)}=\left(a_1 \pmod{r},\cdots, a_n \pmod{r}\right) \in \Z_r^n$. For a list $U=\left(u_1,u_2,\cdots u_t\right)$ over $\Z_m^n$, let $U^{\left(r\right)}=\left(u_1^{\left(r\right)},u_2^{\left(r\right)},\cdots u_t^{\left(r\right)}\right)$.
Also, if $u^{\left(r\right)}$ is constant for all $u \in U$, we say $U^{\left(r\right)}$ is constant.
Similarly, for a $t \times t$ matrix $M$ over $\Z_m$, define $M^{\left(r\right)}$ to be the $t \times t$ matrix over $\Z_r$ such that $M^{\left(r\right)}\left(j,k\right)=M\left(j,k\right)\pmod{r}$ for all $1 \leq j,k \leq t$.
\end{define}

We will also need the following definitions.

\begin{define}[Bucket $B_r\left(w,A\right)$]Let $A \subseteq \Z_m^n$ be a list. For any $w \in \mathbb{Z}_r^n$,
we denote by $B_r\left(w,A\right)=\left(a \in A:a^{\left(r\right)}=w\right)$ the sub-list of elements of $A$ which
are equal to $w$ modulo $r$.
\end{define}

\begin{define}[Matrix $P_{U,V}$]\label{matrix}Let  $U=\left(u_1, u_2, \cdots u_t\right)$ and $V=\left(v_1, v_2, \cdots v_t\right)$ be lists over $\Z_m^n$. We let $P_{U,V}$  be the $t \times t$ matrix over $\Z_m$ defined by $P_{U,V}\left(i,j\right)=\langle u_i, v_j \rangle$ for $1 \leq i,j \leq t$.
\end{define}

We will use the following lemma from \cite{DGY} mentioned informally in the introduction.
\begin{lem}\label{lem-dgy1}\cite[Theorem 21]{DGY} For any positive integer $n$ and prime $p$, $\MV\left(p,n\right) \leq 1+{n+p-2  \choose p-1}$.
\end{lem}

\subsection{Probability Distributions}

\begin{define}For a distribution $\mu$ over a finite set $\Omega$, we write $X \sim \mu$ to denote a random variable $X$ drawn according to $\mu$. We will also treat $\mu$ as a function $\mu: \Omega \mapsto [0,1]$ such that $\mu(x) = \Pr[ X = x]$. For a list  $A$ over $\Omega$, $x \sim A$ denotes a point sampled as per the uniform distribution on $A$ (taking repetitions into account).
\end{define}

\begin{define}[Statistical distance between distributions]  Let $\mu_1$ and $\mu_2$ be two distributions over a finite set $\Omega$. The {\em statistical distance} (or simply distance) between $\mu_1$ and $\mu_2$, denoted  $\Delta\left(\mu_1,\mu_2\right)$, is defined as \[\Delta\left(\mu_1, \mu_2\right)=\frac{1}{2}\sum_{x \in \Omega}\left|\mu_1\left(x\right)-\mu_2\left(x\right)\right|.  \]
\end{define}

\begin{define}[Collision probability] Given a distribution  $\mu$ over a finite set $\Omega$  the {\em collision probability} of $\mu$, denoted $\cp(\mu)$, is defined as \[\cp\left(\mu\right)=\Pr_{x,y \sim \mu}[x=y]=\sum_{x\in \Omega}\mu\left(x\right)^2.\]
\end{define}

The following two lemmas are standard and their proofs are included, for completeness, in Appendix~\ref{sec-distlem}.
\begin{lem}\label{lem-stat}Let $\mu$ be a distribution over $\Z_m$ and let $\U_m$ denote the uniform distriution over $\Z_m$. If $\Delta\left(\mu,\U_m\right)\geq \epsilon$ then for some $1 \leq j \leq m-1$, \[\left|\E_{x \sim \mu}\left[\left(\omega^j\right)^{x} \right]\right|\geq \frac{2\epsilon}{\sqrt{m}},\] where $\omega=exp\left(2\pi i/m\right)$ is  a primitive root of unity of order $m$.
\end{lem}

\begin{lem}\label{lem-cp}Let $\omega$ be a primitive root of unity of order $m$. Let $\mu_1$ and $\mu_2$ be two probability distributions over $\Z_m^n$. If $ \left|\E_{x \sim \mu_1,y\sim \mu_2} \left[\omega^{\langle x,y \rangle}\right]\right|\geq \epsilon $, then $\cp\left(\mu_1\right)\cp\left(\mu_2\right) \geq \epsilon^2/m^n$.
\end{lem}

\section{Proof of Theorem~\ref{thm-thm1}}\label{sec-buildup}

In this section we prove Theorem~\ref{thm-thm1}, restated here with explicit constants.
\begin{thm}\label{thm-mn2} Let $m \geq 2$,$2 \leq q\leq m$ and $n$ be arbitrary positive integers. Then  \[\MV\left(m,n,q\right)\leq 12q \cdot q^{24\left(1+\log q^{10q} \right)}m^{n/2}.\]
\end{thm}

For the purpose of the proof, we  introduce the following notation that will be used only in this section.
\begin{define}[$\MV_{r_1,r_2}\left(m,n,q\right)$] Let $r_1,r_2$ be integers such that $r_1r_2|m$. We denote by $\MV_{r_1,r_2}\left(m,n,q\right)$ the maximum size of a $q-restricted$ MV family $\left(U,V\right)$  in $\mathbb{Z}_m^n$ satisfying
\begin{itemize}
\item $U^{\left(r_1\right)}$ and $V^{\left(r_2\right)}$ are constants.
\item $\langle u,v \rangle=0 \ \left(mod\ r_1r_2\right)$ for all $u \in U, v \in V$.
\end{itemize}
 Note that $\MV_{1,1}\left(m,n,q\right)=\MV\left(m,n,q\right)$ (with the convention that $x \pmod 1 =0$ for any integer $x$).
\end{define}

Before we go to the proof of Theorem ~\ref{thm-thm2}, we have the following claims.

\begin{claim}\label{wlog}Let $(U,V)$ be a q-restricted matching vector in $\Z_m^n$. Then, without loss of generality, $m$ has at most $q$ prime factors.\end{claim}
\begin{proof} Assume $m=\prod_{i=1}^r p_i^{e_i}$ with possible $r>q$. Let $v_1,\cdots,v_q \in Z_m$ be the $q$ possible nonzero values that the inner products $\langle u,v \rangle$ attain. For each $v_j$ there is some prime $p_{i_j}$ where $v_j \neq 0 \left( \mod p_{i_j}^{e_{i_j}}\right)$. So, we can replace $m$ with just $\prod_{j=1}^q p_{i_j}^{e_{i_j}}$ and discard all primes other than $p_{i_1}, \cdots ,p_{i_q}$.
\end{proof}

\begin{claim}\label{highorder}
If $N$ has $r$ prime factors, then $|\{x \in Z_N: order(x)<N/S\}| \leq N/S \cdot \left(\log S\right)^r$.\end{claim}
\begin{proof} Assume $N=\prod_{i=1}^r p_i^{e_i}$. An element $x$ with $order(x) \leq N/S$ is divisible by some $\prod_{i=1}^r p_i^{f_i} \geq S$. Let $T= \{(f_1, \cdots,f_r): \prod p_i^{f_i} \geq S\}$. Define a partial order on $T$ by $(f_1,\cdots,f_r) \leq (f'_1,\cdots,f'_r)$ if $f_i \leq f'_i$. Let $T'$ be a subset of $T$ such that for any $t \in T$ there is $t' \in T'$ such that $t'\leq t$. Note that if $x$ has order $\leq N/S$ then $x$ must be divisible by $\prod_i p_i^{f_i}$ for some $(f_1,\cdots,f_r)$ in $T'$. So, the number of elements of order $< N/S$ is at most $N|T'|/S$. We can bound the size of $T'$ as follows: any element $f_i$ is between 0 and $\log_{p_i} S$, since clearly if $f_i$ is larger we can reduce $f_i$ by one. So, $|T'| \leq \prod_{i=1}^r (\log S / \log p_i) <= (log S)^r$. 
\end{proof}

The proof of Theorem~\ref{thm-mn2} will follow immediately from the following two lemmas, which will be proved below.

\begin{lem}\label{lem-maincs}Let $m=r_1r_2r_3$ where $r_1,r_2,r_3$ are arbitrary positive integers such that $r_3\geq 2$. Let $q \geq 2, t \geq 12q$ and $n$ be arbitrary positive integers. Let $\left(U,V\right)$ be a $q-restricted$ matching vector family in $\Z_m^n$ with $\left|\left(U,V\right)\right|=t$ such that
\begin{itemize}
\item $U^{\left(r_1\right)}$ and $V^{\left(r_2\right)}$ are constants.
\item $\langle u,v \rangle=0 \ \left(mod\ r_1r_2\right)$ for all $u \in U, v \in V$.
\end{itemize}
Then, there exists $s|r_3$ with $s \geq \max\{2, r_3/q^{10q}\}$ and a $q-restricted$  matching vector family $\left(U',V'\right)\subseteq \left(U,V\right)$ such that  
$\left|\left(U',V'\right)\right| \geq  s^{-n/2}q^{-24}t$ where
\begin{itemize}
\item $\langle u',v' \rangle=0 \ \left(mod\ r_1r_2s\right)$ for all $u' \in U', v' \in V'$.
\item Either $U'^{\left(r_1s\right)}$ is constant or $V'^{\left(r_2s\right)}$ is constant.
\end{itemize}
\end{lem}
\bnote{t was unused earlier, now linked to size}

Applying Lemma~\ref{lem-maincs} iteratively we can prove the following bound.

\begin{lem}\label{lem-help}$\MV_{r_1,r_2}\left(m,n,q\right) \leq   12q \cdot q^{24\log \frac{m}{r_1 r_2}}\left(\frac{m}{r_1r_2}\right)^{n/2}$.
\end{lem}

Given Lemma~\ref{lem-help} and Lemma~\ref{lem-maincs}, we now show how to deduce Theorem~\ref{thm-mn2}.

\begin{proof}[Proof of Theorem~\ref{thm-mn2}] Observe that for any matching vector family $\left(U,V\right)$ in $\Z_m^n$, $U^{\left(1\right)}$ and $V^{\left(1\right)}$ are constants and $\langle u,v \rangle=0 \ \left(mod\ 1\right)$ for all $u \in U, v \in V$. Thus, $\MV\left(m,n,q\right)=\MV_{1,1}\left(m,n,q\right)$. 
\emph{Case 1: $m \leq q^{10q}$}. Applying Lemma \ref{lem-help}, we get $\MV\left(m,n,q\right)=\MV_{1,1}\left(m,n,q\right) \leq  12q \cdot q^{24\log q^{10q}}\left(m\right)^{n/2} \leq 12q \cdot q^{24(1+\log q^{10q})}\left(m\right)^{n/2}$.

\emph{Case 2: $m > q^{10q}$}. By Lemma~\ref{lem-help}, we know that for $s\geq m/q^{10q}$, $\MV_{1,s}\left(m,n,q\right)\leq 12q \cdot q^{24\log \frac{m}{s}}\left(\frac{m}{s}\right)^{n/2}\leq 12q \cdot q^{24\log q^{10q}}\left(\frac{m}{s}\right)^{n/2}$. Similarly, we have for $s\geq m/4q$, $\MV_{s,1}\left(m,n,q\right)\leq 12q \cdot q^{24\log q^{10q}}\left(\frac{m}{s}\right)^{n/2}$.

Now, suppose there is a $q-restricted$ MV family $(U,V)$ in $\Z_m^n$ of size $t>12q \cdot q^{24\left(1+\log q^{10q} \right)}m^{n/2}$. Applying Lemma~\ref{lem-maincs} with $r_1=r_2=1$, we get a $q-restricted$ MV family $(U',V')\subseteq (U,V)$ of size $t' \geq  s^{-n/2}q^{-24}t >q^{24\log q^{10q}}\left(\frac{m}{s}\right)^{n/2}$ where $s\geq m/q^{10q}$ such that
\begin{itemize}
\item $\langle u',v' \rangle=0 \ \left(mod\ s\right)$ for all $u' \in U', v' \in V'$.
\item Either $U'^{\left(s\right)}$ is constant or $V'^{\left(s\right)}$ is constant.
\end{itemize}
But, by the previous paragraph, we have for $s\geq m/q^{10q}$, $\MV_{s,1}\left(m,n,q\right)$ and $\MV_{1,s}\left(m,n,q\right)$ are at most $12q \cdot q^{24\log q^{10q}}\left(\frac{m}{s}\right)^{n/2}$. This leads to a contadiction. 
\end{proof}

\subsection{Proof of Lemma \ref{lem-maincs}}

By assumption we have that $\langle u,v \rangle=0 \ \left(mod\ r_1r_2\right)$ for all $u \in U, v \in V$. So, we can consider $\frac{\langle u,v \rangle}{r_1r_2} \in \Z_{r_3}$. Also, by hypothesis, the inner products $\frac{\langle u,v \rangle}{r_1r_2}$ occupy $q'\leq q$ residues in $\Z_{r_3}$. We have that
\begin{itemize}
\item For $1 \leq i \leq t$, $\frac{\langle u_i,v_i \rangle}{r_1r_2}=0 \ \left(mod\ r_3\right)$ since $\langle u_i,v_i \rangle=0 \pmod{m}$.
\item For $1 \leq i,j \leq t$, $i \neq j$, $\frac{\langle u_i,v_j \rangle}{r_1r_2} \neq 0 \ \left(mod\ r_3\right)$ since $\langle u_i,v_j \rangle \ne 0 \pmod{m}$.
\end{itemize}

Let $\mu$ denote the distribution over $\Z_{r_3}$ defined by $\frac{\langle u_i,v_j \rangle}{r_1r_2} \mod r_3$ where $u_i, v_j$ are drawn independently and uniformly from $U,V$ respectively. 

\emph{Case 1: $4q' \geq r_3$}. Observe that $\mu$ outputs $0$ only when $i=j$. Therefore, $\Pr[\mu=0]=1/t \leq 1/12q' \leq 1/3r_3$. On the other hand, $\Pr[\U_{r_3}=0]=1/r_3$. This implies that $\Delta\left(\mu,\U_{r_3}\right) \geq 1/3r_3$.
Thus, applying Lemma \ref{lem-stat} with $\omega=exp\left(2\pi i/r_3\right)$, we get that for some $1 \leq j \leq r_3-1$, \[\left|\E_{x \sim \mu}\left[\left(\omega^j\right)^{x} \right]\right|\geq \frac{2}{3r_3\sqrt{r_3}} \geq \frac{1}{12q'^{3/2}}.\]
Let $\omega'=\omega^j$ and $ord(\omega')$ (the order of $\omega'$) be $s=r_3/gcd\left(r_3,j\right)$. Also, note that as $j\geq 1$, we have $s \geq 2$. Also, trivially, $s \geq r_3/q'^{10q'} \geq r_3/q^{10q}$.

\emph{Case 2: $4q' <r_3$}. Let $X$ be the random variable that picks a random $0\leq j \leq r_3-1$ and outputs $\left|\E_{x \sim \mu}\left[\left(\omega^j\right)^{x} \right]\right|$. We will now show that with significant probability $X^2 \geq 1/2q'$. First observe that $X\leq 1$. On the other hand, we will show that $E\left[X^2\right]$ is large. To see this, let $Z=\{z_1, \cdots z_{q'}\}$ be the $q'$ residues forming the support of $\mu$. Also, for $1 \leq i \leq q'$, let $\alpha_i\stackrel{def}{=}\mu(z_i)$. Then, 
\begin{eqnarray*}
\E_j\left[X^2\right]&=&\E_j\left[\sum_{1 \leq i,i' \leq q'}\alpha_i \alpha_{i'} \omega^{j\left(z_i-z_{i'}\right)}\right]\\
&=&\sum_{1 \leq i \leq q'}\alpha_i^2\\
&\geq &1/q'
\end{eqnarray*}

Therefore, we claim that $\Pr[X^2\geq 1/2q']\geq 1/2q' \geq 1/2q$. If not, then 
\begin{eqnarray*}
E_j\left[X^2\right]&=&\Pr[X^2\geq 1/2q']E_j\left[X^2|X^2 \geq 1/2q'\right]+\Pr[X^2 <1/2q']E_j\left[X^2|X^2 <1/2q'\right]\\
&<& 1/2q'  +  1/2q'\\
&=&1/q'
\end{eqnarray*}
which is a contradiction.

By the above, we already have that there exists some $\omega'$ such that $\left|\E_{x \sim \mu}\left[\left(\omega'\right)^{x} \right]\right| \geq 1/\sqrt{2q'}$ and $ord\left(\omega'\right)\geq 2$ since $r_3/2q'>1$ and thus $\omega'$ is not trivial.

Now, we shall show the existence of $\omega'$ of much higher order provided $r_3>q'^{10q'}$. By Claim \ref{highorder}, for $S=q'^{10q'}$ and $N=r_3$, and noting that $r_3$ has atmost $q$ prime factors by Claim \ref{wlog}, we have 
\[Pr_j[ord\left(\omega^j\right)\leq r_3/S]\leq 1/4q'\]

Thus, with probabilty at least $1/2q'-1/4q'=1/4q'$, a random $j$ satisfies 
\begin{itemize}
\item $\left|\E_{x \sim \mu}\left[\left(\omega^j\right)^{x} \right]\right| \geq 1/\sqrt{2q'} \geq \frac{1}{12q^{3/2}}$
\item $s=ord\left(\omega^j\right)  \geq r_3/S$
\end{itemize}
Also, as $r_3/4q'>1$ the above two conditions are true for some $j \neq 0$.

Now, we combine the above two cases as follows.
Let $\omega'=\omega^j$ and $\eps = \frac{1}{12q^{3/2}}$. We have shown by the above case-by-case analysis that 
\begin{itemize}
\item $\left|\E_{x \sim \mu}\left[\left(\omega'\right)^{x} \right]\right|\geq \epsilon$
\item $s=ord(\omega')$ is such that $s \geq \max\{2,r_3/q^{10q}\}$
\end{itemize}

Using the Cauchy-Schwartz inequality twice we get
\begin{eqnarray*}
&&\left|\E_{u \sim U, v \sim V}\left[\left(\omega'\right)^{\langle u,v \rangle/r_1r_2}\right]\right|\geq \epsilon\\
&\implies& \left|\E_{u,\tilde{u} \sim U, v \sim V}\left[\left(\omega'\right)^{\langle u-\tilde{u},v \rangle/r_1r_2}\right]\right|\geq \epsilon^2\\
&\implies& \left|\E_{u,\tilde{u} \sim U, v,\tilde{v} \sim V}\left[\left(\omega'\right)^{\langle u-\tilde{u},v-\tilde{v} \rangle/r_1r_2}\right]\right|\geq \epsilon^4\\
&\implies& \left|\E_{u,\tilde{u} \sim U, v,\tilde{v} \sim V}\left[\left(\omega'\right)^{\langle \left(u-\tilde{u}\right)/r_1,\left(v-\tilde{v}\right)/r_2 \rangle}\right]\right|\geq \epsilon^4.
\end{eqnarray*}
We need to explain the last expression. Since by assumption $U^{\left(r_1\right)}$ and $V^{\left(r_2\right)}$ are constants, $\left(u-\tilde{u}\right)/r_1 \in \Z_{m}^n$ and $\left(v-\tilde{v}\right)/r_2 \in \Z_{m}^n$ are well defined. Thus, we can fix $\tilde{u}$ and $\tilde{v}$ by an averaging argument such that
\[\left|\E_{u \sim U, v \sim V}\left[\left(\omega'\right)^{\langle \left(u-\tilde{u}\right)/r_1,\left(v-\tilde{v}\right)/r_2 \rangle}\right]\right|\geq \epsilon^4.\]

 Let $U'=\left(u_1',u_2', \cdots u_t'\right), V'=\left(v_1',v_2', \cdots v_t'\right)$ where $u_i'= \left(u_i -\tilde{u}\right)/r_1$ and $v_i'= \left(v_i -\tilde{v}\right)/r_2$. Notice that $U'$ and $V'$ are not assumed to be a MV family (later we will derive from them a MV family). We now define two probability distributions $\mu^{U'}$ and $\mu^{V'}$ over $\Z_s^n$. For each $w \in \Z_s^n$, let $\mu^{U'}\left(w\right)=\left|B_s\left(w,U'\right)\right|/\left|U'\right|$ and $\mu^{V'}\left(w\right)=\left|B_s\left(w,V'\right)\right|/\left|V'\right|$. That is, $\mu^{U'}\left(w\right)$ is the probability
 that $u'^{(s)}=w$ where $u'$ is chosen uniformly in $U'$, and similarly for $\mu^{V'}\left(w\right)$.
 Therefore, since the order of $w'$ is $s$, we have that \[ \left|\E_{w_1 \sim \mu^{U'},w_2 \sim \mu^{V'}}\left[\left(\omega'\right)^{\langle w_1,w_2 \rangle}\right]\right|\geq \epsilon^4.\]

Recalling that $s$ is the order of $\omega'$ and applying Lemma \ref{lem-cp}, we get $\cp\left(\mu^{U'}\right)\cp\left(\mu^{V'}\right) \geq \epsilon^8/s^n$. Therefore, one of $\cp\left(\mu^{U'}\right)$, $\cp\left(\mu^{V'}\right)$, say $\cp\left(\mu^{U'}\right)$, is at least $\epsilon^4/s^{n/2}$. Let $w^*$ be the point of maximum probability mass given by $\mu^{U'}$. Then,
\begin{equation*}
\mu^{U'}\left(w^*\right)=\mu^{U'}\left(w^*\right)\sum_{w \in \Z_s^n}\mu^{U'}\left(w\right)
\geq\sum_{w \in \Z_s^n}\mu^{U'}\left(w\right)^2
=\cp\left(\mu^{U'}\right)
\geq  \epsilon^4/s^{n/2}.
\end{equation*}
Now, $\mu^{U'}\left(w^*\right) \geq  \epsilon^4/s^{n/2}$ means that $\left|\{u \in U: \frac{u-\tilde{u}}{r_1}=w^*\ \left(mod\ s\right)\}\right|\geq t \epsilon^4/s^{n/2} $. Equivalently, \[\bigg|\big\{u \in U: u-\tilde{u}=r_1w^*\ \left(mod\ r_1s\right)\big\}\bigg|\geq t \epsilon^4/s^{n/2}.\]

Let $T'=\left(i:u_i=\tilde{u}+r_1w^*\ \left(mod\ r_1s\right)\right)$. Now, define $U''=\left(u_i:i \in T'\right)$ and $V''=\left(v_i:i \in T'\right)$. Observe that $\left(U'',V''\right)$ is a matching vector family in $\Z_m^n$ such that
\begin{itemize}
\item $U''^{\left(r_1s\right)}$ and $V''^{\left(r_2\right)}$ are constants.
\item $\left|\left(U'',V''\right)\right| \geq t\left(\epsilon^4/s^{n/2}\right)$.
\end{itemize}

The only thing left is to show that $\langle u,v\rangle=0\ \left(mod\ r_1r_2s\right)$ for all $u \in U'',v \in V''$. This may not be true in general. However, we can take a large subset of the matching vector family so that the resulting matching vector family satisfies this condition.
To see this, let $u \in U'', v \in V''$ be arbitrary. Now, $u=r_1s \cdot u'+u_0$ and $v=r_2 \cdot v'+v_0$ where $u', v'$ depend on $u,v$ respectively and $u_0,v_0$ are independent of $u,v$. Then, \[\langle u,v \rangle=r_1r_2s \langle u',v'\rangle+r_1s \langle u',v_0\rangle +r_2 \langle u_0,v'\rangle+ \langle u_0,v_0\rangle.\]
As $u$ varies over $U''$, $\langle u',v_0\rangle$ takes at most $q$ values modulo $r_2$. Hence, $r_1s\langle u',v_0\rangle$ takes at most $q$ values modulo $r_1r_2s$. Therefore, there exist at least $\left(1/q\right)\left|U''\right|$ elements of $U''$ such that $r_1s\langle u',v_0\rangle$ is a constant modulo $r_1r_2s$. We take the corresponding elements from $V''$ to form a matching vector family $\left(U''',V'''\right)\subseteq\left(U'',V''\right)$. We apply another round using the same idea on $U''',V'''$, this time ensuring that $r_2\langle u_0,v'\rangle$ is constant modulo $r_1r_2s$ as $v$ varies over a large fraction of $V'''$. Thus, we end up with $\tilde V$ of size at least $\left(1/q\right)\left|V'''\right|$ such that $r_2\langle u_0,v_i\rangle$ is a constant modulo $r_1r_2s$. We take the corresponding subset $\tilde U$ from $U'''$ so that $(\tilde U, \tilde V)\subseteq \left(U''',V'''\right)$ is a matching vector family. Denote the size of $(\tilde U, \tilde V)$  by $\tilde{t}$. Note that $\tilde{U}=\left(\tilde{u}_1,\cdots ,\tilde{u}_{\tilde{t}}\right), \tilde{V}=\left(\tilde{v}_1,\cdots ,\tilde{v}_{\tilde{t}}\right)$ is a matching vector family in $\Z_m^n$ of size at least $\left(1/q^2\right)t\left(\epsilon^4/s^{n/2}\right)=s^{-n/2}q^{-\left(8+4\log_q \left(12\right)\right)}t\geq s^{-n/2}q^{-\left(8+4\log_2 \left(12\right)\right)}t \geq s^{-n/2}q^{-24}t$.  Also, as $\langle u,v\rangle$ is a constant modulo $r_1r_2s$, for $u \in \tilde{U}, v \in \tilde{V}$, and $\langle \tilde{u}_i,\tilde{v}_i \rangle=0\ \left(mod \ r_1r_2s\right)$, we get that $\langle u,v\rangle =0 \ \left(mod\ r_1r_2s\right)$, for $u \in \tilde{U}, v \in \tilde{V}$. This concludes the proof.
\qed

\subsection{Proof of Lemma \ref{lem-help}}

We prove the lemma by backward induction on $r_1r_2|m$. That is, to prove the claim about $\MV_{r_1,r_2}\left(m,n,q\right)$, we assume the inductive hypothesis for $\MV_{r_1',r_2'}\left(m,n,q\right)$ where $r_1'r_2'>r_1r_2$ and $r_1'r_2'|m$.

\noindent{\bf Base Case.} The base case of $r_1r_2=m$ is trivial. To see this, observe that if $\langle u,v \rangle=0 \ \left(mod\ m\right)$ for all $u \in U, v \in V$, then by the definition of a matching vector family in $\Z_m^n$, the size of such a family cannot exceed $1$. Hence, for $r_1r_2=m$, $\MV_{r_1,r_2}\left(m,n,q\right)=1 \leq  12q\cdot q^{24\log \frac{m}{r_1r_2}}\left(\frac{m}{r_1r_2}\right)^{n/2}$.

\noindent{\bf Inductive Step.} Let $m=r_1r_2r_3$ with $r_1r_2<m$ (that is, $r_3 \geq 2$). By the inductive hypothesis we have  $\MV_{r_1',r_2'}\left(m,n,q\right) \leq 12q \cdot q^{24\log \frac{m}{r_1'r_2'}}\left(\frac{m}{r_1'r_2'}\right)^{n/2}$ for all $r_1',r_2'$ such that $r_1'r_2'>r_1r_2$ and $r_1'r_2'|m$. We need to show that  $\MV_{r_1,r_2}\left(m,n,q\right)\leq  12q \cdot q^{24\log \frac{m}{r_1r_2}}\left(\frac{m}{r_1r_2}\right)^{n/2}$. Suppose this is false, so that there exists a $q-restricted$ matching vector family $\left(U,V\right)$ in $\Z_m^n$ with $U=\left(u_1, \cdots u_t\right),V=\left(v_1,\cdots v_t\right)$ where  $t >12q \cdot q^{24\log \frac{m}{r_1r_2}}\left(\frac{m}{r_1r_2}\right)^{n/2}$ such that
\begin{itemize}
\item $U^{\left(r_1\right)}$ and $V^{\left(r_2\right)}$ are constants.
\item $\langle u,v \rangle=0 \ \left(mod\ r_1r_2\right)$ for all $u \in U, v \in V$.
\end{itemize}
Note that $t \geq 12q$. Therefore, applying Lemma \ref{lem-maincs}, there exists $s|r_3$ with $s \geq 2$ and matching vector family $\left(U',V'\right)\subseteq \left(U,V\right)$ such that $\left|\left(U',V'\right)\right| \geq  s^{-n/2}q^{-24}t$ where \begin{itemize}
\item $\langle u',v' \rangle=0 \ \left(mod\ r_1r_2s\right)$ for all $u' \in U', v' \in V'$.
\item either $U'^{\left(r_1s\right)}$ is constant or $V'^{\left(r_2s\right)}$ is constant.
\end{itemize}
Without loss of generality, we assume that $U'^{\left(r_1s\right)}$ is a constant. Therefore, \begin{eqnarray*}
\left|\left(U',V'\right)\right|& > & s^{-n/2}q^{-24}\cdot 12q\cdot q^{24\log \frac{m}{r_1r_2}}\left(\frac{m}{r_1r_2}\right)^{n/2}\\
&&=12q \cdot q^{24\left(\log \frac{m}{r_1r_2}-1\right)}\left(\frac{m}{r_1r_2s}\right)^{n/2}\\
&&\geq 12q \cdot q^{24\log \frac{m}{r_1r_2s}}\left(\frac{m}{r_1r_2s}\right)^{n/2},
\end{eqnarray*}
where the last inequality used the fact that $s \geq 2$. This however contradicts the inductive hypothesis.
\qed

\section{Matrices over $\Z_m$}\label{sec-zmmatrices}

\paragraph{Notations:} For a $t \times s$ matrix $M$ over $\Z_m$ and for lists $T \subseteq [t], S \subseteq [s]$ the $T \times S$ submatrix of $M$ is the matrix with rows in $T$ and columns in $S$. For $i\in [s]$ and $j \in [t]$ we denote the $i$'th row of $M$ by $M(i:)$ and the $j$'th column by $M(:j)$.

\begin{define}[Span of a set]\label{span}For $A \subseteq \Z_m^n$ let $\spana\left(A\right)$ denote the additive subgroup generated by $A$. We say that a set $A$ spans $u \in \Z_m^n$ if $u \in \spana(A)$.
\end{define}

\begin{define}[Rank of a matrix over $\Z_m$] Let $M$ be a  $t \times t$ matrix over $\Z_m$. Then $\rank\left(M\right)$ is the smallest $r$ such that $M=AB$ where $A$ is an $t \times r$ martrix over $\Z_m$ and $B$ is an $r \times t$ matrix over $\Z_m$.
\end{define}

\begin{define}[Column rank of a matrix over $\Z_m$] Let $M$ be a  $t \times t$ matrix over $\Z_m$. Let $\colspan\left(M\right)$ denote the subgroup of $\Z_m^t$ generated by the columns of $M$. The column rank of $M$ over $\Z_m$ is defined as $$\colrank\left(M\right)=\log_m \left|\colspan\left(M\right)\right|.$$ The column rank is, in general, a real number in the range $[0,t]$.
\end{define}

Since the rank can behave in unexpected ways over $\Z_m$, we make sure to prove some of the basic facts that we will be using later on.
\begin{fact}\label{rankmonotonic}Let $M$ be a  $t \times t$ matrix over $\Z_m$ and let $M'$ be any submatrix of $M$. Then $\colrank\left(M'\right)\leq \colrank\left(M\right)$.\end{fact}
\begin{proof}
	Suppose $M'$ is given by the first $t'$ rows and the first $t''$ columns of $M$.	 We will define an injective map $f:\colspan\left(M'\right) \rightarrow \colspan\left(M\right)$.  Given any $x \in \colspan\left(M'\right)$ we can write $x = \sum_{j=1}^{t''} \alpha_j \cdot M'(:j)$ in some fixed way (there might be several choices of $\alpha_j$). Define $f(x) = \sum_{j=1}^{t''} \alpha_j \cdot M(:j)$. Then, $x$ is clearly the restriction of $f(x)$ to the first $t'$ indices and so the map is injective.	
\end{proof}

\begin{fact}\label{rankmon}Let $M$ be a  $t \times t$ matrix over $\Z_m$ and let $s|m$.  Then $\rank\left(M^{\left(s\right)}\right)\leq \rank\left(M\right)$.\end{fact}
\begin{proof}
	Suppose there exist an $t \times r$ matrix $A$ and an $r \times t$ matrix $B$ over $\Z_m$ such that $M=AB$. Then $M^{\left(s\right)}=A^{\left(s\right)}B^{\left(s\right)}$ and so the rank of $M^{(s)}$ is at most $r$.
\end{proof}

We will need the following claims relating the rank and the column rank of matrices over $\Z_m$.
\begin{claim}\label{clm-rankcolrank} Let $M$ be an $t \times t$ matrix  over $\Z_m$. Then, $$\frac{\rank\left(M\right)}{\log m} \leq \colrank\left(M\right) \leq  \rank\left(M\right).$$
\end{claim}
\begin{proof}
Let $r=\rank\left(M\right)$ and $r'=\colrank\left(M\right)$.
  We first prove that $r' \leq r$. This is equivalent to proving that $\left|\colspan\left(M\right)\right| \leq m^r$. Let $M=AB$ where $A$ is an $t \times r$ martrix over $\Z_m$ and $B$ is an $r \times t$ matrix over $\Z_m$.  Since the columns of $M$ are all in the span of the columns of $A$ we have that the column span of $M$ can contain at most $m^r$ elements.

We now prove that $r' \geq r/\left(\log m\right)$ or, equivalently, $\left|\colspan\left(M\right)\right| \geq 2^r$. Suppose in contradiction that $\left|\colspan\left(M\right)\right| < 2^r$. Take a minimal spanning set $S$ of $\colspan\left(M\right)$ (that is,
a set that spans $\colspan\left(M\right)$ and such that no proper subset of it
does). Suppose $|S| \geq r$  and consider all linear combinations (over $\Z_m$) of elements of $S$ with coefficients in $\{0,1\} \subseteq \Z_m$. Since $\left|\colspan\left(M\right)\right| < 2^r$ there are two distinct $0-1$ linear combinations that map to the same element. This means that there is a linear combination with coefficients in $\{1,-1\}$ of the elements of $S$ that is equal to zero. Since both $1$ and $-1$ are invertible modulo $m$ we can write one of the elements of $S$ as a linear combination of the other elements. This contradicts the minimality of $S$ and so, we must have $|S| < r$. This implies that $\rank(M) < r$, a contradiction, since we can write $M$ as the product of the matrix with columns in $S$ with the matrix of  coefficients giving the columns of $M$.
\end{proof}

\begin{claim}\label{clm-colspan}Let $M$ be an $t \times t$ matrix over $\Z_m$, let  $r=\rank\left(M\right)$. There exists $r'$ columns of $M$ that span the rest of $M'$s
columns such that $r' \leq r\log m$.
\end{claim}
\begin{proof}
Take a minimal spanning set $S$ of the columns of $M$ (that is,
a set that spans all other columns and such that no proper subset of it
spans all columns). If $2^{\left|S\right|} >  m^r$, then $2^{\left|S\right|} >  \left|\colspan(M)\right|$ (by Claim \ref{clm-rankcolrank}) and we proceed as in the proof from Claim \ref{clm-rankcolrank} above. If we look at all the $0-1$ combinations of the columns of $S$, then there are two distinct $0-1$ linear combinations of the columns that map to the same element of $\colspan\left(M\right)$. Thus, let $\sum_i\al_iS\left(:i\right)=\sum_i\beta_iS\left(:i\right)$ where $\al_i \neq \beta_i$ for at least one $i$, say $i_0$. Therefore, we have $\sum_i\left(\al_i-\beta_i\right)S\left(:i\right)=0$. Note that $\left(\al_{i_0}-\beta_{i_0}\right)=\pm 1$ and hence is invertible. This lets us write $S\left(:i_0\right)$ as a linear combinations of the remaining columns contradicting the minimality of $S$.
Thus, $r'=\left|S\right| \leq r\log m$.
\end{proof}

The following claim shows that the column rank behaves similar to rank in terms of subadditivity.
\begin{claim}\label{clm-colrankadd}Let $A, B$ be $t \times t$ matrices over $\Z_m$. Then, $\colrank\left(A+B\right)\leq \colrank\left(A\right)+\colrank\left(B\right)$.
\end{claim}
\begin{proof}
We show that $\left|\colspan\left(A+B\right)\right|\leq \left|\colspan\left(A\right)\right|\left|\colspan\left(B\right)\right|$. Note that $\colspan\left(A+B\right) \subseteq \colspan\left(A\right)+ \colspan\left(B\right) \stackrel{def}{=} \{a+b|a\in \colspan\left(A\right), b \in \colspan\left(B\right)\}$. Therefore, $\left|\colspan\left(A+B\right)\right| \leq \left|\colspan\left(A\right)+ \colspan\left(B\right)\right| \leq \left|\colspan\left(A\right)\right|\left|\colspan\left(B\right)\right|$.
\end{proof}

\begin{claim}\label{clm-coltriangular}Let $M$ be a $2t \times 2t$ matrix over $\Z_m$, such that \[M=\left(\begin{array}{cc}A&0\\\star&B\end{array}\right)\] where $A, B $ and $\star$ are $t \times t$  matrices. Then, $\colrank\left(A\right)+\colrank\left(B\right) \leq \colrank\left(M\right)$.
\end{claim}
\begin{proof}We show that $\left|\colspan\left(A\right)\right|\left|\colspan\left(B\right)\right|\leq \left|\colspan\left(M\right)\right|$. Let $\colspan\left(A\right)=R_1$,
$\colspan\left(B\right)=R_2$, $\colspan\left(M\right)=R$. We define $f:R_1 \times R_2 \rightarrow
R$ and show that $f$ is injective.
Given $r_1 \in R_1$ and $r_2 \in R_2$, let $\alpha_1, \cdots \alpha_t$ and $\beta_1,
\cdots \beta_t$ denote coefficients for linear combinations of the columns of $A$ and $B$ respectively
that give $r_1$ and $r_2$. There might be many such linear combinations but we
fix one for each $r_i$. Then, $f\left(r_1,r_2\right)=\sum_{i=1}^{t}\alpha_iM\left(:i\right)+\sum_{i=t+1}^{2t}\beta_{i-t}M\left(:i\right)$. Now, given a column vector $f\left(r_1,r_2\right) \in R$, we
uniquely identify $r_1$ and $r_2$ as follows. We look at the first $t$ rows and call
it $s_1$. Now $s_1=r_1$ and let $\alpha_1, \cdots \alpha_{t}$ be the linear combination fixed for $r_1$ while
defining $f$. Now, consider $f\left(r_1,r_2\right)-\sum_{i=1}^{t}\alpha_iM\left(:i\right)$ and call the last $t$ rows $s_2$. Note that $s_2=r_2$.
\end{proof}

\begin{claim}\label{clm-rankone}Let $M$ be a $t \times t$ square matrix over $\Z_m$ with zero diagonal entries. If for some $s |m$, $\colrank\left(M^{\left(s\right)}\right)\leq 2$, then there exists at least $t'=t/m^2$ indices such that $M$ restricted to those indices as rows and columns is the all zero matrix modulo $s$.\end{claim}
\begin{proof}As $\colrank\left(M^{\left(s\right)}\right)\leq 2$, it follows that $\left|\colspan\left(M^{\left(s\right)}\right)\right|\leq s^2 \leq m^2$. Hence, $M^{\left(s\right)}$ has at most $m^2$ distinct columns. Therefore, there exists a set of indices $S$ of size $t'\geq t/m^2$ with $S=\{r_1, r_2,  \cdots  r_{t'}\}$ such that all the columns $M^{\left(s\right)}\left(:r_i\right)$ are identical. Also, as the diagonal elements are zero modulo $m$, they are zero modulo $s$. Thus, the $S \times S$ submatrix is the all zero matrix modulo $s$.
\end{proof}

\section{Collision-Free MV families}\label{sec-collfree}

In the proof of Theorem~\ref{thm-thm2} it will be useful to assume that the elements of the MV family do not `collide' when reduced modulo an integer $s$ dividing $m$. In this section we develop the necessary machinery to allow for this assumption. We start by defining a collision free matching vector family.
\begin{define}[Collision free MV family]\label{nocoll}A \emph{collision free matching vector family} $\left(U,V\right)$ in $\mathbb{Z}_m^n$ is a matching vector family such that
for all $s|m, s \geq 2$, all elements of $U$ are distinct modulo $s$, and all elements of $V$ are distinct modulo $s$.
Note that if $\left(U,V\right)$ is a collision free matching vector family, then so is any $\left(U',V'\right)\subseteq \left(U,V\right)$.
\end{define}

\begin{lem}\label{lem-colr}Let $m\geq 2$ be an arbitrary integer. Let $s$ be a divisor of $m$, such that $1<s<m$. Let $\left(U,V\right)$ be a matching vector family in $\mathbb{Z}_m^n$ such that $\langle u,v \rangle=0\ \left(mod\ s\right)$ for all $u \in U, v \in V$. Then, $|(U,V)| \leq \MV\left(m/s,n\log m\right)$.\end{lem}
\begin{proof}Let $U=\left(u_1, u_2, \cdots u_t\right)$ and $V=\left(v_1, v_2, \cdots v_t\right) $. Recall that $P_{U,V}$ is the inner product matrix. We shall write $P_{U,V}$ as $P$ in the rest of the proof for brevity. Let $r=\rank\left(P\right) \leq n$. Hence, by Claim \ref{clm-colspan}, there exists $r' \le r \cdot \log m$ columns of $P$ which span all the columns of $P$.
As each entry of $P$ is a multiple of $s$ we can define a matrix $P'$ over $\Z_{m/s}$ by $P'=\left(1/s\right)P$. We have
\begin{itemize}
\item $P'_{i,i}=0 \quad \forall i$.
\item $P'_{i,j}  \neq 0 \quad  \forall i \neq j$.
\end{itemize}

We next show that the $r'$ columns that span the columns of $P$ also span the columns in $P'$.
Without loss of generality, let the first $r'$ columns of $P$ span the remaining columns of $P$. For any column $j$, let $P\left(:j\right)=\sum_{i=1}^{r'}c_iP\left(:i\right) \pmod{m}$. Since all entries of $P$ are divisible by $s$, we can divide the expression by $s$ and obtain that $P'\left(:j\right)=\sum_{i=1}^{r'}c_i P'\left(:i\right) \pmod{m/s}$.
Hence, we deduce that $r_{P'}=\rank\left(P'\right) \leq r' \leq r\log m \leq n\log m$. This implies that $P'=AB$ for some $t \times r_{P'}$ matrix $A$ and some $r_{P'} \times t$ matrix $B$ over $\Z_{m/s}$. Thus, the rows of $A$ and the columns of $B$ form a matching vector family in $\Z_{m/s}^{r_{P'}}$. Therefore, $t \leq \MV\left(m/s,n\log m \right)$ as claimed.
\end{proof}

\begin{lem}[Bucket Lemma]\label{lem-regularity}For any $m$, let $\left(U,V\right)$ be a matching vector family in $\Z_m^n$. Let $1<s<m$ be any divisor of $m$. Then, for any $w \in \mathbb{Z}_s^n$, $\left|B_s\left(w,U\right)\right| \leq \MV\left(m/s,n\log m\right)$. By symmetry, $\left|B_s\left(w,V\right)\right| \leq \MV\left(m/s,n\log m\right)$.
\end{lem}
\begin{proof}We prove that $\left|B_s\left(w,U\right)\right| \leq \MV\left(m/s,n\right)$. For $U=\left(u_1, u_2, \cdots u_t\right) $, consider any bucket $B_s\left(w,U\right)=U' \ \left(say\right)$. Let $U'=\left(u_{j_1}, u_{j_2}, \cdots u_{j_{t'}}\right)$ where $1 \leq j_1 < j_2 < \cdots j_{t'} \leq t$. Let $V'=\left(v_{j_1}, v_{j_2}, \cdots v_{j_{t'}}\right)$. Now, for any $l,m \in [t']$, $\langle u_{j_l},v_{j_l} \rangle=0 \left(mod \ m\right)$. Therefore, $\langle u_{j_m},v_{j_l} \rangle=0 \left(mod \ s\right)$. By Lemma \ref{lem-colr} on $\left(U',V'\right)$, $t' \leq \MV\left(m/s,n\log m\right)$.
\end{proof}

We use the above lemma repeatedly to obtain a collision free matching vector family.
\begin{lem}\label{lem-nocollision}Let $m\geq 2$ be any positive integer. Suppose there is a matching vector family $\left(U,V\right)$ in $\mathbb{Z}_m^n$. Then, there exists a collision free matching vector family $\left(U',V'\right) \subseteq \left(U,V\right)$  such that $$\left|\left(U',V'\right)\right|  \geq  \frac{\left|\left(U,V\right)\right|}{\left(\prod_{s|m, 1<s<m}	 \MV\left(s,n\log m\right)\right)^2}.$$
\end{lem}
\begin{proof}
	We will get rid of collisions iteratively by repeatedly applying Lemma \ref{lem-regularity}. Let us write the divisors of $m$ in ascending order as $2 \leq s_1 < s_2 <\cdots <  s_l\leq m/2$. Perform the following operation for each $s|m$ starting from the smallest divisor greater than $1$.  For $0 \leq i \leq l$, let $U_i,V_i$ be the matching vector after stage $i$ with $U_0=U$ and $V_0=V$. Now suppose that we have $U_i,V_i$ after the $i$'th stage such that there is no collision modulo $s_j$ in $U_i$ for $1 \leq j \leq i$. The $\left(i+1\right)$'th stage is performed as follows. Let us construct $U_{i+1},V_{i+1}$ from $U_i,V_i$ to ensure no collision among the elements of $U_{i+1}$ modulo $s_{i+1}$ as well.
For each $w \in \Z_{s_{i+1}}^n$, by  Lemma \ref{lem-regularity}, $\left|B_{s_{i+1}}\left(w,U_i\right)\right|\leq \MV\left(m/s_{i+1},n\log m\right)$. Pick one element from each bucket in $U_i$ and the corresponding matching vector from $V_i$ to form $\left(U_{i+1},V_{i+1}\right)$. Thus, $\left|\left(U_{i+1},V_{i+1}\right)\right| \geq \left|U_i\right|/\MV\left(m/s_{i+1},n\log m\right)$. We end up with matching vector family $U_l,V_l$ such that $\left|\left(U_l,V_l\right)\right|  \geq  \frac{\left|\left(U,V\right)\right|}{\prod_{s|m, 1<s<m}	\MV\left(m/s,n\log m\right)}$ and $U_l$  is collision free. We repeat the same process this time pruning $V_l$ in order to make it collision free as well. Thus, eventually we end up with a collision free matching vector family $\left(U_l',V_l'\right) \subseteq \left(U,V\right)$ such that \[\left|\left(U_l',V_l'\right)\right| \geq \frac{\left|\left(U,V\right)\right|}{\left(\prod_{s|m, 1<s<m}	\MV\left(m/s,n\log m\right)\right)^2}=\frac{\left|\left(U,V\right)\right|}{\left(\prod_{s|m, 1<s<m}	\MV\left(s,n\log m\right)\right)^2}.\]
\end{proof}

\section{Proof of Theorem~\ref{thm-thm2}}\label{sec-main}
Before proceeding with the proof we give yet another definition.

\begin{define}\label{Di}Let $A, B \subseteq \mathbb{Z}_m^n$ be twin-free lists (or sets). Let $\omega$ be a primitive root of unity of order $m$. The duality measure of $A,B$ with respect to $\omega$ is defined as
\[D_{\omega}\left(A,B\right)= \left|\E_{a \sim A, b \sim B}\left[\omega^{\langle a,b \rangle}\right]\right|.\] Notice that, if $\omega \neq 1$, $D_{\omega}(A,B)=1$ implies that there is some $c \in \Z_m$ such that all the entries of the inner product matrix $P_{A,B}$ equal $c$. We often refer to such submatrices as monochromatic rectangles. 
\end{define}

The following is an easy consequence of Lemma \ref{lem-stat}.
\begin{lem}\label{lem-bias}Let $\left(U,V\right)$ be a MV family in $\mathbb{Z}_m^n$ of size $t \geq 3m$ and let $\omega=exp\left(2\pi i/m\right)$ be a primitive root of unity of order $m$. Then there exists some $1 \leq j \leq m-1$ such that \[D_{\omega^j}\left(U,V\right)\geq \frac{2}{3m^{3/2}}.\]
\end{lem}
\begin{proof}
Let $\mu$ be the random variable which chooses $u \in U$ and $v \in V$ randomly and outputs $\langle u,v \rangle$ and let $\U_m$ be the uniform distribution over $\Z_m$. Now, $\Delta\left(\mu,\U_m\right)\geq \left(1/2\right)\left(\Pr[\U_m=0]-\Pr[\mu=0]\right)= \left(1/2\right)\left(1/m-1/t\right)\geq 1/3m$ as $t \geq 3m$.
By Lemma \ref{lem-stat}, for some $1 \leq j \leq m-1$,  \[\left|\E_{x \sim \mu}\left[\left(\omega^j\right)^{x} \right]\right|\geq \frac{2}{3m^{3/2}}.\]
Thus, we have $\left|\E_{u \sim U, v \sim V}\left[\left(\omega^j\right)^{\langle u,v \rangle} \right]\right|\geq \frac{2}{3m^{3/2}}$ as claimed.
\end{proof}

An important ingredient in the proof of Theorem~\ref{thm-thm2} is the following lemma, referred to in the introduction as the `sub-matrix lemma' which is a generalization of a result of \cite{SLZ}.

\begin{lem}[Sub-Matrix Lemma]\label{lem-slz}Let $s,m,n \geq 2$ where $s$ divides $m$, and let $\omega$ be a primitive root of unity of order $s$. Let $A,B \subset \Z_s^n$ be two twin-free lists satisfying $D_{\omega}\left(A,B\right)  \geq \frac{2}{3m^{3/2}}$. Let $\rank\left(P_{A,B}\right)=r \geq 2$. Then assuming Conjecture \ref{pfr} (PFR conjecture), there exist lists $A'\subseteq A, B' \subseteq B$ such that $D_{\omega}\left(A', B'\right) = 1$, where
$|A'| \geq  2^{-c\left(m\right)r/\log r}\left|A\right|$,	$|B'|\geq 2^{-c\left(m\right)r/\log r} \left|B\right|$ for some constant $c\left(m\right)$ which depends only on $m$.
\end{lem}
Without loss of generality, we can assume $c(m) \geq 1$ above (it will be convenient to assume it in the proof of Theorem~\ref{thm-thm2}). In other words, we can replace the $c(m)$ above by $\max\{c(m),1\}$. We postpone the proof
of Lemma~\ref{lem-slz} to Section~\ref{sec-mono} and proceed now with the proof of Theorem~\ref{thm-thm2}.

We restate Theorem~\ref{thm-thm2} here for convenience and with the explicit function $d(m)$.
\begin{thm}\label{main}Let $n,m \geq 2$ be arbitrary positive integers. Then, assuming Conjecture \ref{pfr} (PFR conjecture), we have \[ \MV\left(m,n\right) <  2^{d\left(m\right)n/\log n},\]
where $d\left(m\right)=1200c\left(m\right)m^{6\log m}$ and $c\left(m\right)$ is as in  Lemma \ref{lem-slz}.
\end{thm}
\begin{proof}  We prove the theorem by induction on the number of (not necessarily distinct) prime factors of $m$.
\snote{removed explicit definition of $\alpha$ since these are not used in the proof}

\paragraph{Choice of $d\left(m\right)$.}
Let  $d,d_1,d_2,d_3:\Z^+ \rightarrow \R$ be functions and $d_4$ be a constant. We want the following conditions to be satisfied for all $m,n \geq 2$.
\begin{enumerate}
\item $d\left(m\right),d_1\left(m\right),d_2\left(m\right),d_3\left(m\right)$ are monotonically increasing in $m$
\item $\left(2n\right)^m \leq 2^{d\left(m\right)n/\log n}$
\item $\left(2m\right)^m \leq 2^{d\left(m\right)n/\log n}$
\item $d\left(m\right) \geq d\left(m/2\right) \cdot 4 m \log m$
\item $-d_2\left(m\right)+\left(1/2\right)d\left(m\right)>d\left(m/2\right)\log m$
\item $2^{\left(1/2\right)d\left(m\right)n/\log n} \geq 3m2^{d_2\left(m\right)n/\log n}$
\item $d_2\left(m\right)n/\log n \geq 2\log m+d_3\left(m\right)n/\log n$
\item $d_3\left(m\right) \geq d_1\left(m\right) \cdot d_4 \cdot m \log m$
\item $d_4 \geq 300$
\item $d_1\left(m\right)\geq 2c\left(m\right)$
\item $d_2 \geq d_3+1$
\end{enumerate}

It can be verified that the following choice for the functions meets the above conditions.
\begin{itemize}
\item $d\left(m\right)=1200 \cdot c\left(m\right) \cdot m^{6\log m}$
\item $d_1\left(m\right)=2 \cdot c\left(m\right)$
\item $d_2\left(m\right)=602 \cdot c\left(m\right) \cdot m \log m$
\item $d_3\left(m\right)=600 \cdot c\left(m\right)\cdot m\log m$
\item $d_4=300$
\end{itemize}

We shall explicitly mention which conditions of the above functions are being used in different parts of the proof.

\paragraph{Base Case.}
The base case is where $m=p$ is prime. Lemma \ref{lem-dgy1} implies that $\MV\left(p,n\right) \leq 1+{n+p-2  \choose p-1}<\left(2\max\{n,p\}\right)^p$. If we show $\left(2n\right)^p \leq 2^{d\left(p\right)n/\log n}$ and $\left(2p\right)^p \leq 2^{d\left(p\right)n/\log n}$ we will be done. Indeed, by the choice of $d\left(m\right)$ (Condition $2$ and $3$) both of the above will hold.

\paragraph{Inductive Case.}
Let $n \geq 2, m \geq 2$ be arbitrary positive integers. Suppose, by induction, that $\MV\left(s,n\right) < 2^{d\left(s\right)n/\log n}$ for all $s|m, s<m$. We need to show that, assuming Conjecture \ref{pfr}, \[ \MV\left(m,n\right) < 2^{d\left(m\right)n/\log n}\]
Suppose not. That is, there exists a matching vector family $\left(U,V\right)$  of size $t \geq 2^{d\left(m\right)n/\log n}$.
First, we shall apply Lemma \ref{lem-nocollision} to $\left(U,V\right)$ to obtain a large enough collision free matching vector family $\left(U',V'\right)$.

\paragraph{\hspace{5mm}A large collision free matching vector family.}
We show that $\left|\left(U',V'\right)\right|\geq 2^{\left(1/2\right)d\left(m\right)n/\log n}$.
Let $\left|\left(U',V'\right)\right|=t'$. Observe that by Lemma \ref{lem-nocollision}, the inductive hypothesis and the monotonicity of $d\left(m\right)$ (Condition $1$),  $t' \geq 2^{d\left(m\right)n/\log n-2m \cdot d\left(m/2\right) \cdot n \log m/\log n}$ where we have used a loose upper bound of $m$ for the number of factors of $m$. Now,
\begin{eqnarray*}
&&t'\geq   2^{\left(1/2\right)d\left(m\right)n/\log n}\\
&if &d\left(m\right)n/\log n-2m \cdot d\left(m/2\right) \cdot n \log m/\log n \geq \left(1/2\right)d\left(m\right)n/\log n\\
&\Leftrightarrow&d\left(m\right) \geq d\left(m/2\right)\cdot 4m \log m
\end{eqnarray*}
which is satisfied by the choice of $d\left(m\right)$ (Condition $4$).

\paragraph{\hspace{5mm}Two key claims.}

We will need two claims from which the inductive claim follows easily. We shall provide proofs to these claims after the proof of the inductive claim.
\begin{claim}\label{clm-iterate}Let $\left(U,V\right)$ be a collision free matching vector family in $\Z_m^n$ with $\left|\left(U,V\right)\right| \geq 3m$ and $\colrank\left(P_{U,V}^{\left(s'\right)}\right) > 2$ for all $s'|m, s'\geq 2$. Then, for some $s|m, s\geq 2$, there exists a collision free matching vector family $\left(U',V'\right)\subseteq \left(U,V\right)$ in $\Z_m^n$ satisfying
\begin{itemize}
\item $\left|\left(U',V'\right)\right|\geq 2^{-d_1\left(m\right)r_s/\log r_s}\left|\left(U,V\right)\right|$ where $r_s=\rank\left(P_{U,V}^{\left(s\right)}\right)$.
\item Either $\colrank\left(P_{U',V'}^{\left(s\right)}\right) \leq \left(3/4\right)\colrank\left(P_{U,V}^{\left(s\right)}\right)$ or $\colrank\left(P_{U',V'}^{\left(s\right)}\right) \leq 2$.
\end{itemize}
\end{claim}

\begin{claim}\label{clm-allzero}Let $\left(U,V\right)$ be a collision free matching vector family in $\Z_m^n$ such that $\left|\left(U,V\right)\right|\geq 3m \cdot 2^{d_2\left(m\right)n/\log n}$ . Then, there exists a collision free matching vector family $\left(U',V'\right)\subseteq \left(U,V\right)$ in $\Z_m^n$ satisfying
\begin{itemize}
\item $\left|\left(U',V'\right)\right|\geq 2^{-d_2\left(m\right)n/\log n}\left|\left(U,V\right)\right|$.
\item  $P_{U,V}^{\left(s\right)}$ is the all zero matrix for some $s|m, s \geq 2$.
\end{itemize}
\end{claim}

Let us proceed with the proof of the inductive claim assuming these two claims.
We have a collision free matching vector family $\left(U',V'\right)$ with $\left|\left(U',V'\right)\right|\geq 2^{\left(1/2\right)d\left(m\right)n/\log n} \geq 3m \cdot 2^{d_2\left(m\right)n/\log n}$. (Condition $6$ satisfied by the choice of $d\left(m\right), d_2\left(m\right)$)   Applying Claim \ref{clm-allzero}, there exists a collision free matching vector family $\left(U'',V''\right)\subseteq \left(U',V'\right)\subseteq \left(U,V\right)$ in $\Z_m^n$ satisfying
\begin{itemize}
\item $\left|\left(U'',V''\right)\right|\geq 2^{-d_2\left(m\right)n/\log n}2^{\left(1/2\right)d\left(m\right)n/\log n}$.
\item  $P_{U'',V''}^{\left(s\right)}$ is the all zero matrix for some $s|m, s \geq 2$.
\end{itemize}

By the choice of $d\left(m\right)$, it can be verified that $-d_2\left(m\right)+\left(1/2\right)d\left(m\right)>d\left(m/2\right)\log m$ (Condition $5$). Thus, $\left|\left(U'',V''\right)\right| > 2^{d\left(m/2\right)n\log m/\log n}$.

We now show that this is enough to get a contradiction. If $s=m$, we have $\left|\left(U'',V''\right)\right|\leq 1$ as $\left(U'',V''\right)$ is a matching vector family in $\Z_m^n$. If $s<m$, by Lemma \ref{lem-colr} and the inductive hypothesis, we have $\left|\left(U'',V''\right)\right|\leq 2^{d\left(m/s\right)n\log m/\log \left(n \log m\right)} \leq 2^{d\left(m/2\right)n\log m/\log n}$  by monotonicity of $d\left(m\right)$ (Condition $1$). Thus, irrespective of $s$, $\left|\left(U'',V''\right)\right|\leq 2^{d\left(m/2\right)n\log m/\log n}$ which is a contradiction.
This completes the proof.
\end{proof}

\vspace{5mm}
\paragraph{Proof of Claim \ref{clm-iterate}:}
Let $\left|\left(U,V\right)\right|=t \geq 3m$. Let $\omega$ be a root of unity of order $m$. By Lemma \ref{lem-bias}, for some $1\leq j \leq m-1$ , $D_{\omega^j}\left(U,V\right)\geq \frac{2}{3m^{3/2}}$. Note that $s=m/gcd\left(m,j\right)$ is the order of $\omega'=\omega^j$. Observe that $s|m, s\geq 2$ as $1\leq j \leq m-1$. Recall from the statement of the claim that $r_s=rank\left(P_{U,V}^{\left(s\right)}\right)$. Thus, by the collision free property of $\left(U,V\right)$, \[D_{\omega'}\left(U^{\left(s\right)},V^{\left(s\right)}\right)=\left|\E_{u \sim U^{\left(s\right)}, v \sim V^{\left(s\right)}}\left[\left(\omega'\right)^{\langle u,v \rangle}\right]\right|=\left|\E_{u \sim U, v \sim V}\left[\left(\omega'\right)^{\langle u,v \rangle}\right]\right|=D_{\omega'}\left(U,V\right) \geq \frac{2}{3m^{3/2}}.\] Applying Lemma \ref{lem-slz} on $U^{\left(s\right)},V^{\left(s\right)}$ with $\omega'$ a primitive root of unity of order $s$, we can get an $\left(R \times S\right)$ submatrix of $P_{U,V}$ with $|R|=|S| \geq 2^{-c\left(m\right)r_s/\log r_s}t$. (we can make $|R|=|S|$ as throwing away rows and columns from a monochromatic rectangle still keeps it monochromatic) Let $T=R \cap S$. We divide our analysis to two cases: either $|T|>|R|/2$ or $|T| \leq |R|/2$.
In both cases, we shall exhibit a matching vector family as required in the statement of the claim.

\emph{Case 1:  $|T|>|R|/2$. }
For $U=\left(u_1, u_2, \cdots u_t\right)$, $V=\left(v_1, v_2, \cdots v_t\right)$, let $U'=\left(u_j|j\in T\right)$ and $V'=\left(v_j|j\in T\right)$, and $P'=P_{U',V'}$. Now, as $P'^{\left(s\right)}$ is monochromatic, and $\langle u_j,v_j\rangle=0\ \left(mod\ s\right)$ for $j \in T$, we have $\langle u',v' \rangle =0 \left(mod\ s\right)$ for all $u' \in U', v' \in V'$. Observe that \begin{itemize}
\item $\left|\left(U',V'\right)\right| \geq 2^{-1-c\left(m\right)r_s/\log r_s}t \geq 2^{-2c\left(m\right)r_s/\log r_s}t\geq 2^{-d_1\left(m\right)r_s/\log r_s}t$ (by the choice of $d_1\left(m\right)$, Condition $10$)
\item $\colrank\left(P_{U',V'}^{\left(s\right)}\right) =0 \leq 2$
\end{itemize}
This finishes Case $1$.

\emph{Case 2:  $|T| \leq |R|/2$. }
Let $R'=R \setminus T$ and $S'=S \setminus T$. Note that $R' \cap S'=\emptyset$ and $|R'|=|S'|$. Consider the $R'\cup S' \times R'\cup S'$ submatrix of $P_{U,V}$. Call it $P'$. Note that \[P'^{\left(s\right)}=\left(\begin{array}{cc}P_1'&C\\\star&P_2'\end{array}\right)\] where $P_1'$ and $P_2'$ are the $R' \times R'$ and the $S' \times S'$ submatrices of $P_{U,V}^{\left(s\right)}$ respectively and $C$ is monochromatic. We add a matrix of column rank at most $1$ to $P'^{\left(s\right)}$ to yield $P''^{\left(s\right)}$ which is the same as $P'^{\left(s\right)}$ except that $C$ is replaced by the all zero block matrix. Thus, \[P''^{\left(s\right)}=\left(\begin{array}{cc}P_1'&0\\\star&P_2'\end{array}\right)\] Note that by Claim \ref{clm-colrankadd}, $\colrank\left(P''^{\left(s\right)}\right) \leq \colrank\left(P'^{\left(s\right)}\right)+1$.  Now, using Claim \ref{clm-coltriangular}, $\colrank\left(P_1'\right)+\colrank\left(P_2'\right)\leq \colrank\left(P'^{\left(s\right)}\right)+1 \leq \colrank\left(P_{U,V}^{\left(s\right)}\right)+1 \leq  \left(3/2\right)\colrank\left(P_{U,V}^{\left(s\right)}\right)$ as $\colrank\left(P_{U,V}^{\left(s\right)}\right) >2$. Therefore, one of $P_1', P_2'$, say $P_1'$ satisfies $\colrank\left(P_1'\right)\leq \left(3/4\right)\colrank\left(P_{U,V}^{\left(s\right)}\right)$. Construct the matching vector family $\left(U',V'\right)$ as follows. Let $U'=\left(u_j|j\in R'\right)$ and $V'=\left(v_j|j\in R'\right)$. Again, observe that \begin{itemize}
\item $\left|\left(U',V'\right)\right| \geq 2^{-1-c\left(m\right)r_s/\log r_s}t \geq 2^{-2c\left(m\right)r_s/\log r_s}t\geq 2^{-d_1\left(m\right)r_s/\log r_s}t$ (by the choice of $d_1\left(m\right)$, Condition $10$).
\item $\colrank\left(P_{U',V'}^{\left(s\right)}\right)  \leq (3/4)\colrank\left(P_{U,V}^{\left(s\right)}\right)$.
\end{itemize}
This completes the proof of Case $2$.

\qed

\paragraph{Proof of Claim \ref{clm-allzero}:}
We will use Claim~\ref{clm-iterate} iteratively. For this, we first set up some notations.
\paragraph{The setup.}
Define a sequence of collision free matching vector families for $i=0,\ldots,z$.
\begin{itemize}
\item $\left(U,V\right)=\left(U_0,V_0\right), \left(U_1,V_1\right) \cdots$
\item Let $t_i=\left|\left(U_i,V_i\right)\right|$.
\item Each step $i$ has label $s_i|m$ (this label will be given by Claim \ref{clm-iterate}).
\item Let $cr_i:\Z^+ \rightarrow \R$ be defined by \[cr_i\left(s\right)=\colrank\left(P_{U_i,V_i}^{\left(s\right)}\right).\]
\item Let $r_i:\Z^+ \rightarrow \Z$ be defined by \[r_i\left(s\right)=\rank\left(P_{U_i,V_i}^{\left(s\right)}\right).\]
\end{itemize}

\paragraph{Invariants.}
We will show how to go from step $i$ to step $i+1$. We stop after stage $z$ when $cr_z\left(s\right) \leq 2$ for some $s|m, s\geq 2$. We shall maintain the following invariants for $0 \leq i \leq z-1$.
\begin{itemize}
\item $\left(U_{i+1},V_{i+1}\right)\subseteq \left(U_{i},V_{i}\right)$ and hence is a collision free matching vector family in $\Z_m^n$.
\item $t_{i+1} \geq 2^{-d_1\left(m\right)r_i\left(s_i\right)/\log r_i\left(s_i\right)}t_i$.
\item $cr_{i+1}\left(s_i\right)\leq \left(3/4\right)cr_i\left(s_i\right)$ or $cr_{i+1}\left(s_i\right)\leq 2$.
\item $cr_{i+1}\left(s'\right)\leq cr_i\left(s'\right)$ for all $s'|m$.
\end{itemize}

\paragraph{Step $i$ $\rightarrow$ Step $i+1$.}
We state a claim that we will prove below.
\begin{claim}\label{clm-summation}$\sum_{i=0}^{z-1}d_1\left(m\right)r_i\left(s_i\right)/\log r_i\left(s_i\right)\leq d_3\left(m\right)n/log n$.\end{claim}

In order to apply Claim \ref{clm-iterate}, we need to satisfy $t_i \geq 3m$. Observe that by Claim \ref{clm-summation},
\begin{eqnarray*}
t_i\geq t_z &\geq& t_0\prod_{j=0}^{z-1}2^{-d_1\left(m\right)r_j\left(s_j\right)/\log r_j\left(s_j\right)}\\ &\geq& 2^{-d_3\left(m\right)n/\log n}t_0 \\&\geq& 3m \cdot 2^{-d_3\left(m\right)n/\log n+d_2\left(m\right)n/\log n}\geq 3m,
\end{eqnarray*}
(by the choice of $d_2\left(m\right), d_3\left(m\right)$ in Condition $11$).  Apply Claim \ref{clm-iterate} to $\left(U_{i},V_{i}\right)$ to get label $s_i$ for step $i$ and $\left(U_{i+1},V_{i+1}\right)\subseteq \left(U_{i},V_{i}\right)$. The first three invariants are maintained by the statement of Claim \ref{clm-iterate}. The last invariant follows from Fact \ref{rankmonotonic}. Note that by the inequality we just established, $t_z\geq 2^{-d_3\left(m\right)n/\log n}t_0$. Also, by the stopping condition, $cr_z\left(s'\right) \leq 2$ for some $s'|m, s'\geq 2$. Thus, applying Claim \ref{clm-rankone}, we get another matching vector family $\left(U',V'\right)\subseteq \left(U_z,V_z\right)\subseteq \left(U,V\right)$ such that
\begin{itemize}
\item $\left|\left(U',V'\right)\right|\geq t_z/m^2 \geq 2^{-2\log m-d_3\left(m\right)n/\log n}\left|\left(U,V\right)\right|\geq 2^{-d_2\left(m\right)n/\log n}\left|\left(U,V\right)\right|$ (Condition $7$ satisfied by the choice of $d_2\left(m\right)$ and $d_3\left(m\right)$).
\item $P_{U',V'}^{\left(s'\right)}$ is the all zero matrix.
\end{itemize}
This finishes the Proof of Claim \ref{clm-allzero}.

\emph{Proof of Claim \ref{clm-summation}:}
Let $t_s$ be the number of steps with label $s$. Note that as the column rank modulo $s$ goes down by a factor of at least $3/4$ each time we are in a step labeled $s$, it is easy to see that $t_s \leq \log_{4/3}cr_0\left(s\right) \leq \log_{4/3}n$. We shall rely on the monotonic increasing nature of $x/\log x$ when $x \geq e$. As $cr_i\left(s\right)>2$, by Claim \ref{clm-rankcolrank}, $r_i\left(s\right)\geq cr_i\left(s\right)> 2$ which means $r_i(s) \geq 3 > e$ as the rank is always an integer.
We thus have
\begin{eqnarray*}
&&\sum_{i=0}^{z-1}d_1\left(m\right)\frac{r_i\left(s_i\right)}{\log r_i\left(s_i\right)}\\
&\leq & d_1\left(m\right)\log m\sum_{i=0}^{z-1}\frac{cr_i\left(s_i\right)}{\log cr_i\left(s_i\right)}\ \ \text{(by Claim \ref{clm-rankcolrank}) and monotonicity of $x/\log x$ as discussed above}\\
&\leq & d_1\left(m\right)\log m\sum_{s|m, s\geq 2}\sum_{j=1}^{\lfloor\log_{4/3}n\left(s\right)\rfloor}\left(\frac{cr_0\left(s\right)}{\left(4/3\right)^{j-1}\log \left(cr_0\left(s\right)/\left(4/3\right)^{j-1}\right)}\right)\\
&\leq & d_1\left(m\right)\log m\sum_{s|m, s\geq 2}d_4cr_0\left(s\right)/\log cr_0\left(s\right)\ \ \text{(by Claim \ref{clm-sum} and Condition $9$ satisfied by $d_4$)}\\
&\leq & d_1\left(m\right)\log m\sum_{s|m, s\geq 2}d_4n/\log n\ \ \text{(as $cr_0\left(s\right)\leq r_0\left(s\right)\leq r_0\left(m\right)\leq n$, by  Claim \ref{clm-rankcolrank} and Fact \ref{rankmon})}\\
&\leq & d_4d_1\left(m\right)m\left(\log m\right)n/\log n\\
&\leq & d_3\left(m\right)n/\log n\ \ \text{(by the choice of $d_3\left(m\right)$, Condition $8$)}
\end{eqnarray*}
This completes the proof.
\qed

\section{Monochromatic rectangles from low rank matrices}\label{sec-mono}
In this section we prove Lemma~\ref{lem-slz} (the Sub-Matrix Lemma). We begin with some preliminary definitions. The following is a standard result in algebra and can be find in any introductory text.
\begin{thm}[Fundamental Theorem of finitely generated abelian groups]\label{fund}Every finitely generated abelian group $G$ is isomorphic to a direct product of cyclic groups of prime power order and an infinite cyclic group. More precisely, \[G \cong \Z^n \times \Z_{q_1} \times \Z_{q_2} \cdots \times \Z_{q_r}\]
where $q_i$'s are prime powers with $q_1 \leq q_2 \cdots \leq  q_r$. The decomposition is unique after applying this ordering on $q_i$'s. If the group $G$ is finite, then $n=0$.\end{thm}

We will use the following two definitions regarding sumsets.
\begin{define}[Difference Set] For $A \subseteq \Z_m^n$ define its difference set as $A-A{=}\{a-a'|a,a'\in A\}$.
\end{define}

\begin{define}[$rep_S\left(x\right)$]For any $S \subseteq \Z_m^n$ and $x \in \Z_m^n$, $rep_S\left(x\right)$ is the number of different representations of $x$ as an expression of the form $s-s'$ where $s,s' \in S$.
\end{define}
\bnote{changed the above two defn by making minus}

Next, we define the $\epsilon$-spectrum of $B$ with respect to a primitive root of unity of order $m$.
\begin{define}[Spectrum]For $B \subseteq \Z_m^n$, and $\epsilon \in [0,1]$, the $\epsilon$-spectrum of $B$ with respect to $\omega$, a primitive root of unity of order $m$, is the set \[\Spec_{\epsilon}\left(B\right)=\left\{x \in \mathbb{Z}_m^n:\left|\E_{b \sim B}\left[\omega^{\langle x,b \rangle}\right]\right| \geq \epsilon\right\}.\]
When $\omega$ is implicit in the context, we will drop the phrase "with respect to $\omega$".
\end{define}

We start by proving the following lemma which is a generalization of  a lemma from \cite{SLZ}.
\begin{lem}\label{lem-szring}Let  $A,B \subseteq \mathbb{Z}_m^n$ be sets. Let $\omega$ be a primitive root of unity of order $m$. If $A \subseteq \Spec_{\epsilon}\left(B\right)$, then there exist sets $A'\subseteq A, B'\subseteq B$, such that $|A'| \geq |A|/m$ and $|B'| \geq \epsilon^2\frac{|A|}{|span\left(A\right)|}|B|$ such that $D_{\omega}\left(A',B'\right)=1$.
\end{lem}

\begin{proof}
We start by setting up some notations. Let $W=\spana\left(A\right)$ be the subgroup of $\Z_m^n$ spanned by $A$. By Theorem \ref{fund}, there exists an isomorphism $\tau:\prod_{i=1}^r\Z_{q_i}\rightarrow W$. Let $\C=\prod_{i=1}^r\Z_{q_i}$ and note that we can think of elements of $\C$ as vectors with integer coordinates where the $i$'th coordinate is in $\Z_{q_i}$. Let  $e_1,e_2, \cdots e_r \in \C$ where $e_i$ is the vector that has $1$ in the $i$'th coordinate and $0$ everywhere else. Given $x \in \C$, $\exists \alpha_1, \cdots \alpha_r$, with $\alpha_i \in \Z_{q_i}$ such that \[x=\sum_{i=1}^r\alpha_ie_i.\]
Then $\tau\left(x\right)=\sum_{i=1}^r\alpha_i\tau\left(e_i\right)$. Let $v_i=\tau\left(e_i\right)$ for $1 \leq i \leq r$. We can think of the $v_i$'s as a basis of $W$. Therefore, for $\alpha=\left(\alpha_1, \alpha_2, \cdots \alpha_r\right) \in \C$ we have $\tau\left(\alpha\right)=\sum_{i=1}^r\alpha_iv_i$. Let $$\Theta=\{\left(\beta_1, \cdots \beta_r\right) \in \Z_m^r|\exists u \in \Z_m^n \text{ such that } \forall i, \beta_i=\langle v_i,u\rangle\}.$$ 

\begin{claim}\label{clm-basis}For $1 \leq i \leq r$, $q_iv_i=0^n\ \left(mod \ m\right)$.\end{claim}
\begin{proof}Let $x =0^r \in \C$. Now $\tau\left(x\right)=0^n\ \left(mod\ m\right)$. Note that $x$ can also be written as $x=q_ie_i$. Applying $\tau$ on both sides, we get $\tau\left(x\right)=q_iv_i$. Thus, $q_iv_i=0^n\ \left(mod\ m\right)$.
\end{proof}

\begin{claim}\label{clm-beta}For $\beta \in \Theta$, $1 \leq i \leq r$, $q_i\beta_i=0\ \left(mod \ m\right)$.\end{claim}
\begin{proof}As $\beta \in \Theta$, there is a $u \in \Z_m^n$ such that $\forall i$, $\beta_i=\langle v_i,u\rangle$. Then, $q_i\beta_i=q_i\langle v_i,u\rangle=0\ \left(mod\ m\right)$ by Claim \ref{clm-basis}.
\end{proof}

For $\alpha \in \C, \beta \in \Theta$ we define their inner product $\langle \alpha, \beta \rangle \in \Z_m$ by 
considering $\alpha_i \in \{0,\ldots,q_i-1\}, \beta_i \in \{0,\ldots,m-1\}$, taking the inner product over the integers and then reducing the result modulo $m$. This is indeed an inner product by Claim~\ref{clm-beta}.

\begin{claim}\label{clm-zero}Given $\beta \in \Theta \setminus \{0\}$, \[\sum_{a \in W}\omega^{\langle \tau^{-1}\left(a\right),\beta \rangle}=\sum_{\alpha \in \C}\omega^{\langle \alpha,\beta \rangle}=0.\]
\end{claim}
\begin{proof}
Let $\beta_i \neq 0$. Then $\sum_{\alpha \in C}\omega^{\langle \alpha,\beta \rangle}=0$ whenever $\sum_{j=0}^{q_i-1}\omega^{j\beta_i}=0$. Now, $\sum_{j=0}^{q_i-1}\omega^{j\beta_i}=\frac{\omega^{q_i\beta_i}-1}{\omega^{\beta_i}-1}$. This is well defined because $\omega$ is of order $m$ and $\beta_i\neq 0$. The claim now follows from Claim \ref{clm-beta} which makes the expression zero.
\end{proof}

With the above setup in place, we can now proceed with the proof of Lemma~\ref{lem-szring}.
For $\beta \in \Theta$,
  define \[S_{\beta}=\{x \in \Z_m^n|\langle v_i,x\rangle=\beta_i, 1\leq i\leq r\}.\] Denoting $\mu\left(\beta\right)=\Pr_{b \in B}[b \in S_{\beta}]$, we observe that $\cup_{\beta \in \Theta}\left(B\cap S_{\beta}\right)=B$. Hence, $\sum_{\beta \in \Theta}\mu\left(\beta\right)=1$. For $a \in W$, define $h\left(a\right)=\E_{b \in B}\left[\omega^{\langle a,b \rangle}\right]$. If $a=\sum_{i=1}^r\alpha_iv_i$ then
\begin{eqnarray*}
h\left(a\right)&=&\E_{b \in B}\left[\omega^{\langle a,b \rangle}\right]\\
&=&\E_{b \in B}\left[\omega^{\langle \sum_{i=1}^r\alpha_iv_i,b \rangle}\right]\\
&=&\sum_{\beta \in \Theta}\mu\left(\beta\right)\omega^{\langle \alpha ,\beta \rangle}\\
&=&\sum_{\beta \in \Theta}\mu\left(\beta\right)\omega^{\langle \tau^{-1}\left(a\right) ,\beta \rangle}.
\end{eqnarray*}
We will prove upper and lower bounds for the sum $\sum_{a \in A}\left|h\left(a\right)\right|^2$. On the one hand,
\begin{eqnarray*}
\sum_{a \in A}\left|h\left(a\right)\right|^2&\geq&\frac{1}{|A|}\left(\sum_{a \in A}\left|h\left(a\right)\right|\right)^2 \ \text{(Cauchy Scwartz inequality)}\\
&\geq&\frac{1}{|A|}\left(\sum_{a \in A}\epsilon\right)^2 \ \text{($A \subseteq \Spec_{\epsilon}\left(B\right)$ implies $\left|h\left(a\right)\right|\geq \epsilon$ )}\\
&\geq & |A|\epsilon^2.
\end{eqnarray*}
On the other hand,
\begin{eqnarray*}
\sum_{a \in A}\left|h\left(a\right)\right|^2&\leq  & \sum_{a \in W}\left|h\left(a\right)\right|^2\\
&=&\sum_{a \in W}\sum_{\beta \in \Theta, \beta' \in \Theta'}\mu\left(\beta\right)\mu\left(\beta'\right)\omega^{\langle \tau^{-1}\left(a\right),\beta - \beta' \rangle}\\
&=&\sum_{\beta,\beta'  \in \Theta,}\mu\left(\beta\right)\mu\left(\beta'\right)\sum_{a \in W}\omega^{\langle \tau^{-1}\left(a\right),\beta - \beta' \rangle}\\
&=&\sum_{\beta \in \Theta}\mu\left(\beta\right)^2|W| \qquad \text{(Claim \ref{clm-zero})}\\
&\le& |W| \max_{\beta \in \Theta}\{\mu\left(\beta\right)\}.
\end{eqnarray*}

Now, combining the upper and lower bounds,
$\max_{\beta \in \Theta}\{\mu\left(\beta\right)\}  \geq \epsilon^2\frac{|A|}{|W|}$. Thus, there exists a $\beta \in \Theta$ such that $\mu\left(\beta\right)\geq \epsilon^2\frac{|A|}{|W|}$. This means that the subset $B'=B \cap S_{\beta}$ is of size at least $\epsilon^2\frac{|A|}{|W|}|B|$. Now, any $a \in A$ can be written as $a=\sum_{i=1}^r \alpha_i v_i$, and for $b \in B'$, the inner product $\langle a,b \rangle=\langle \alpha,\beta \rangle$ is independent of $b$. Now, for $i \in [m]$, let $A_i \subseteq A$ be such that for $a \in A_i$, for all $b \in B'$, $\langle a,b \rangle=\langle \alpha,\beta \rangle =i$. Now there exists some $A_i$, call it $A'$, of size at least $|A|/m$ such that $\langle a',b' \rangle=i$ for all $a'\in A',b' \in B'$, that is, $D_{\omega}\left(A',B'\right)=1$ and this proves the lemma.
\end{proof}

We continue along the lines of \cite{SLZ} and prove the following lemma.
\begin{lem}\label{lem-slzlemma}Suppose the twin free lists $U, V \subseteq \mathbb{Z}_m^n$ satisfy $D_{\omega}\left(U, V\right) \geq \epsilon$̨ where $\omega$ is a primitive root of unity of order $m$. Also, let $\rank\left(P_{U,V}\right)=r$. Then assuming Conjecture \ref{pfr}, for every $K > 1$, letting $\ell = r/\log_m K$, there exist lists $U'\subseteq U, V'\subseteq V$ such that $D_{\omega}\left(U', V'\right) = 1$, and
$|U'| \geq  poly_m\left(\frac{\left(\epsilon/2\right)^{2^{\ell}}}{rK}\right)\left(2mr\right)^{-\ell}|U|$,	$|V'| \geq poly_m\left(\frac{\left(\epsilon/2\right)^{2^{\ell}}}{rK}\right)m^{-\ell}|V|$.
\end{lem}
\begin{proof}
Let $U=\left(u_1, \cdots u_t\right)$ and $V=\left(v_1, \cdots v_t\right)$. Since $P_{U,V}$ has rank $r$ there exists a $t\times r$ matrix $U_M$ and $r \times t$ matrix $V_M$ so that $U_MV_M=P_{U,V}$. Thus if we let  $A$ denote the rows of $U_M$ and $B$ denote the columns of $V_M$, then $A,B \subseteq \Z_m^r$. The proof does not care about the order of elements and hence we now consider $A,B$ which are sets. Note that $|A|=|B|=t$  and if $A=\left(a_1, \cdots a_t\right)$ and $B=\left(b_1, \cdots b_t\right)$ then $\langle a_i,b_j \rangle=\langle u_i, v_j \rangle$ for $1\leq i,j \leq t$. Thus, $D_{\omega}\left(U, V\right) \geq \epsilon$ implies $D_{\omega}\left(A, B\right) \geq \epsilon$.
Following \cite{SLZ} consider a sequence of constants $\epsilon_1=\epsilon/2$,  $\epsilon_2=\epsilon_1^2/2$, $\epsilon_3=\epsilon_2^2/2$, $\cdots$
and a sequence of sets $A_1=A \cap \Spec_{\epsilon_1}\left(B\right)$ and $A_{i} \subseteq \left(A_{i-1}-A_{i-1}\right) \cap \Spec_{\epsilon_i}\left(B\right)$. The way the subsets are chosen for $A_i$'s will be made precise shortly.
Now by the pigeonhole principle, there exists a minimal index $\ell\leq r/\log_m K$ such that $|A_{\ell+1}| \leq K|A_{\ell}|$. \bnote{cardinality sign outside $A_{l+1}$ fixed} To give a precise definition of the $A_i$'s , we have the following. Let $A_1=A \cap \Spec_{\epsilon/2}\left(B\right)$. For $i\geq 2$,  assuming $\epsilon_{i-1}$ and $A_{i-1}$, let $j_i$ be the the integer index which maximizes the size of \[\{\left(a,a'\right) \in A_{i-1}\times A_{i-1}|a-a' \in \Spec_{\epsilon_i}\left(B\right) \ and \  m^{j_i}\leq rep_{A_{i-1}}\left(a-a'\right) \leq m^{j_i+1}\},\]
and let \[A_i=\{a-a' | a,a' \in A_{i-1}, a-a' \in \Spec_{\epsilon_i}\left(B\right) \ and \ m^{j_i}\leq rep_{A_{i-1}}\left(a-a'\right) \leq m^{j_i+1}\}.\]

\begin{claim}\label{clm-claim24}For $i=1$ we have $|A_1| \geq \left(\epsilon̨/2\right)|A|$. For $i>1$ we have $\Pr_{a,a' \in A_{i-1}}[a-a' \in A_i] \geq \epsilon_i/r$ and additionally
$|A_i| \geq \frac{\epsilon_i}{m^{j_i+1}r}|A_{i-1}|^2$.
\end{claim}
\begin{proof}
The case of $i=1$ follows from Markov inequality. For larger $i$, we show that \[\Pr_{a,a' \in A_{i-1}}[a-a' \in \Spec_{\epsilon_i}\left(B\right)] \geq \epsilon_i.\]
This follows from the fact that \[\epsilon_{i-1}^2 \leq  \left|\E_{b \in B, a \in A_{i-1}}\left[\omega^{\langle a,b \rangle}\right]\right|^2 \leq \E_{b \in B}\left|\E_{a \in A_{i-1}}\left[\omega^{\langle a,b \rangle}\right]\right|^2 = \E_{a,a'\in A_{i-1}}\E_{b \in B}\left[\omega^{\langle a-a',b \rangle}\right].\] Now applying Markov inequality we get that $\Pr_{a,a' \in A_{i-1}}[a-a' \in \Spec_{\epsilon_i}(B)] \geq \epsilon_i=\epsilon_{i-1}^2/2$. Now selecting $j_i$ as in the construction gives that $\Pr_{a,a' \in A_{i-1}}[a-a' \in A_i] \geq \epsilon_i/r$.
\bnote{2 changes above, one, changed the last inequality to equality in the chain of inequalities. two, changed the phrase 'atleast .. fraction sum to an element...' to the corresponding difference statement}

To prove the second part of the lemma, observe that by the above, we have shown that $$\left|\{\left(a,a'\right) \in A_{i-1}\times A_{i-1}|a-a' \in A_i\}\right| \geq \frac{\epsilon_i}{r}|A_{i-1}|^2.$$ Also, by construction \bnote{spelling fixed} of $A_i$, since every $x \in A_i$ can be represented as $x=a-a'$ with $a,a' \in A_{i-1}$ in at most $m^{j_i+1}$ ways, we have that $|A_i| \geq \frac{\epsilon_i}{m^{j_i+1}r}\left|A_{i-1}\right|^2$.
This completes the proof.
\end{proof}

Below we will use  the following additive-combinatorics lemma.
\begin{thm}[\cite{BalSze,Gow}]\label{BSG}There exists an absolute constant $c>0$ such that the following holds. Let $A$ be any arbitrary subset of an abelian group $G$. Let $S \subseteq G$ be such that $|S|\leq C|A|$. If $\Pr_{a,a' \in A}[a-a' \in S] \geq 1/C$, then there exists a subset $A' \subseteq A$ such that $|A'|\geq \frac{|A|}{C^c}$ and $|A'-A'| \leq C^c|A|$.
\end{thm}
\bnote{fixed BS and G citations, BS was pointing to a different paper here and G had the wrong year}

Now we come to the main claim.
\begin{claim} For $i=\ell,\ell-1, \cdots 1$ there exist subsets $A_i' \subseteq A_i$, $B_i' \subseteq B$ such that $D_{\omega}\left(A_i',B_i'\right)=1$ and \[|A_i'|\geq \alpha_i|A_i|\]
and
\[|B_i'|\geq \beta_i|B|\]
where $\alpha_i=poly_m\left(\frac{\epsilon_{\ell+1}}{rK}\right)\left(2mr\right)^{-\left(\ell-i\right)}\left(\prod_{j=i}^{\ell}\epsilon_{j+1}\right), \beta_i=poly_m\left(\frac{\epsilon_{\ell+1}}{rK}\right)m^{-\left(\ell-i\right)}$
\end{claim}
\bnote{to be fixed: the product in alpha should run from i to l-1 because technically in the base case we just have a poly $\eps_{l+1}$ contribution. The current summation gives another $\eps_{l+1}$ factor which is unnecessary. We can fix it and thus the calculation in the next page for i=1 or leave it unchanged, as the extra $\eps_{l+1}$ can be absorbed in the poly($\eps_{l+1}$) term}
\noindent{\bf Base Case.}
The base case of $i=\ell$ is proved by an application of the Balog-Szemeredi-Gowers theorem followed by Conjecture \ref{pfr} followed by Lemma \ref{lem-szring}. To see this, we know that $|A_{\ell+1}| \leq K|A_{\ell}|$ and $\Pr_{a,a' \in A_{\ell}}[a-a' \in A_{\ell+1}] \geq \epsilon_{\ell+1}/r$. Hence by Theorem \ref{BSG} (with $C=\frac{rK}{\epsilon_{\ell+1}}$), there exists a set $A_{\ell}''\subseteq A_{\ell}$ such that $|A_{\ell}''| \geq poly\left(\frac{\epsilon_{\ell+1}}{rK}\right)\left|A_{\ell}\right|$ and $|A_{\ell}'' - A_{\ell}''|\leq poly\left(\frac{rK}{\epsilon_{\ell+1}}\right)\left|A_{\ell}''\right|$. Now by Conjecture \ref{pfr} applied to $A_{\ell}''$, there exists a set $A_{\ell}'''\subseteq A_{\ell}''$ such that $\left|A_{\ell}'''\right| \geq poly_m\left(\frac{\epsilon_{\ell+1}}{rK}\right)\left|A_{\ell}''\right|$ and $\left|\spana\left(A_{\ell}'''\right)\right|\leq m|A_{\ell}''| = poly_m\left(\frac{rK}{\epsilon_{\ell+1}}\right)\left|A_{\ell}'''\right|$. (Note the extra factor of $m$ in front of $|A_{\ell}''|$ as we get a coset of size $|A_{\ell}''|$ and its span incurs an additional factor of $m$)
Also, as $A_{\ell}''' \in \Spec_{\epsilon_{\ell}}\left(B\right)$, applying Lemma \ref{lem-szring} to $A_{\ell}'''$ and $B$, we get $A_{\ell}' \subseteq A_{\ell}'''$ and $B_{\ell}' \subseteq B$ such that $D_{\omega}\left(A_{\ell}',B_{\ell}'\right)=1$, $|A_{\ell}'|\geq poly_m\left(\frac{\epsilon_{{\ell}+1}}{rK}\right)\left|A_{\ell}\right|$ and $|B_{\ell}'|\geq poly_m\left(\frac{\epsilon_{\ell+1}}{rK}\right)\left|B\right|$.
This completes the base case. Let us come to the inductive case.
\qed

\noindent{\bf Inductive Case.}
Suppose the statement is true for $i$ and let us argue for $i-1$. Let $G=\left(A_{i-1},E\right)$ be the graph whose vertices are the elements in $A_{i-1}$ and $\left(a,a'\right)$ is an edge if $a-a' \in A_i'$. Now,
\begin{eqnarray*}
|E|&\geq&m^{j_i}|A_i'|\\
&\geq & m^{j_i}\alpha_i|A_i| \ \ \text{(inductive hypothesis)}\\
&\geq & m^{j_i}\alpha_i\frac{\epsilon_i}{m^{j_i+1}r}|A_{i-1}|^2 \ \ \text{(Claim \ref{clm-claim24})}\\
&=&2\alpha_{i-1}|A_{i-1}|^2
\end{eqnarray*}
Now the graph has at least $2\alpha_{i-1}|A_{i-1}|^2$ edges and $|A_{i-1}|$ vertices and  therefore has a connected component of size at least $2\alpha_{i-1}|A_{i-1}|$ vertices. Let us call these vertices $A_{i-1}''$. Let $\tilde{a}$ be any element of $A_{i-1}''$. Partition $B_i'$ into $B_{i,j}'$ for $0 \leq j \leq m-1$ such that all elements of $B_{i,j}'$ have inner product $j$ with $\tilde{a}$. Let $B_{i-1}'=B_{i,j_1}$ be the largest of them. Note that $|B_{i-1}'| \geq |B_i'|/m$.
By assumption $D_{\omega}\left(A_i',B_i'\right)=1$. Hence, $D_{\omega}\left(A_i',B_{i-1}'\right)=1$. Therefore, for some $j_2$,  $\langle a,b \rangle = j_2$ for all $a \in A_i'$ and $b \in B_{i-1}'$. Now, in the connected component obtained above, whenever $a,a' \in A_{i-1}''$ are neighbours, $\langle a-a',b \rangle=j_2$ for $b \in B_{i-1}'$. Thus, starting with $\tilde{a}$ as the anchor and propagating throughout the connected component, we can classify the vertices in $A_{i-1}''$ based on the inner product it has with all elements in $B_{i-1}'$, which is either $j_1$ or $j_2-j_1$. Pick the larger set and call it $A_{i-1}'$. Hence, $D_{\omega}\left(A_{i-1}',B_{i-1}'\right)=1$. Thus, $|A_{i-1}'|\geq |A_{i-1}''|/2 \geq  \alpha_{i-1}|A_{i-1}|$ and $B_{i-1}'\geq |B_i'|/m \geq \frac{\beta_i}{m}|B|= \beta_{i-1}|B|$. This completes the inductive case.
\qed

Put $i=1$ in the above claim. Also observe that as $\epsilon_{j+1}=\epsilon^{2^j}/2^{2^j-1}\geq \left(\epsilon/2\right)^{2^j}$. Thus, $\epsilon_{\ell+1}\geq \left(\epsilon/2\right)^{2^{\ell}}$ and $\prod_{j=1}^{\ell}\epsilon_{j+1}\geq \left(\epsilon/2\right)^{2^{\ell+1}}$ there exist $A' \subseteq A_1 \subseteq A$, $B' \subseteq B$, such that $|A'| \geq  poly\left(\frac{\left(\epsilon/2\right)^{2^{\ell}}}{rK}\right)\left(2mr\right)^{-\ell}\left|A\right|$ and $ |B'| \geq poly\left(\frac{\left(\epsilon/2\right)^{2^{\ell}}}{rK}\right)m^{-\ell}\left|B\right|$. Observing that the lower bounds grow weaker with increasing $\ell$,and that $\ell\leq  \ell'=r/\log_m K$ we get $|A'| \geq  poly\left(\frac{\left(\epsilon/2\right)^{2^{\ell'}}}{rK}\right)\left(2mr\right)^{-\ell'}\left|A\right|$ and $|B'| \geq poly\left(\frac{\left(\epsilon/2\right)^{2^{\ell'}}}{rK}\right)m^{-\ell'}\left|B\right|$ where $\ell'=r/\log_m K$. Therefore, if we take the list $U'\subseteq U$ (corresponding to $A'\subseteq A$)  and $V'\subseteq V$ (corresponding to $B'\subseteq B$) then as $\langle a_i,b_j \rangle=\langle u_i, v_j \rangle$ the statement of the lemma follows. This completes the proof of Lemma~\ref{lem-slzlemma}
\end{proof}

We can now prove the Sub-Matrix Lemma, Lemma~\ref{lem-slz}.
\paragraph{Proof of Lemma~\ref{lem-slz}:}
Set $K=s^{4r/\log r}, \ell=\frac{\log r}{4}, \epsilon=1/2m^{3/2}$ while applying Lemma \ref{lem-slzlemma} over $\Z_s$. We get $|A'|\geq \delta_s |A|$, $|B'|\geq \delta_s |B|$ where \begin{eqnarray*}
\delta_s&=&poly_s\left(\frac{1}{m^{r^{1/4}}}\right)2^{-c_1\left(s\right)r/\log r}\ \ \text{(for some constant $c_1\left(s\right)$ depending only on $s$)}\\
&\geq & poly_m\left(\frac{1}{m^{r^{1/4}}}\right)2^{-c_1\left(s\right)r/\log r}
\end{eqnarray*}
Now let $c_2\left(m\right)=\max_{s|m, s\geq 2}\{c_1\left(s\right)\}$. Thus, $\delta_s \geq poly_m\left(\frac{1}{m^{r^{1/4}}}\right)2^{-c_2\left(m\right)r/\log r}\geq 2^{-c\left(m\right)r/\log r}$ for some constant $c$ that depends only on $m$. \qed

\bibliographystyle{alpha}
\bibliography{LDC_PIR_biblio}


\appendix

\section{A Calculation}
\begin{claim}\label{clm-sum}Let $b>1, n \geq 2$ be arbitrary integers. Then \[\sum_{i=1}^{\lfloor \log_b n\rfloor}\frac{1}{b^{i-1}\log \left(n/b^{i-1}\right)}\leq f(b)/\log n\] where $f(b)=\frac{10b}{b-1}+\frac{10}{\log b}+\frac{16e}{\log^2b}$. When $b=4/3$, $f(b)<300$.
\end{claim}
\begin{proof}
We divide the summation into two parts. The first part consists of the first  $\lfloor\log_b \log n \rfloor$ terms and the second part consists of the remaining terms.

In the first part, $\frac{1}{b^{i-1}\log n/b^{i-1}}\leq \frac{1}{b^{i-1}0.1\log n}$ whenever $n \geq 2$ and hence the first part summation is bounded from above by $\frac{10b}{(b-1)\log n}$.

In the second part of the summation, we use the monotonicity of $x \log \left(n/x\right)$. The function increases with $x$ as long as $x\leq n/e$. Therefore, for terms with $b^{i-1}\leq n/e$, the maximum value of each summand is given by substituting $i=\log_b\log n$ which gives an upper bound of $\frac{1}{0.1\log^2 n}$. The remaining terms corresponding to $n/b \geq b^{i-1} > n/e$ (note that these extra terms arise only if $b<e$) can be analysed as follows. Observe that each summand in that range can be upperbounded by $\frac{e}{n\log b}$. Therefore, we have at most $\log_b n$ terms each at most $\frac{10}{\log^2 n}+\frac{e}{n\log b}$. Thus, the second part of the summation is bounded from above by $\log_b n\left(\frac{10}{\log^2 n}+\frac{e}{n\log b}\right)$.
\begin{eqnarray*}
\log_b n\left(\frac{10}{\log^2 n}+\frac{e}{n\log b}\right)&=&\frac{10}{\log b}\frac{1}{\log n}+\frac{e}{\log^2b}\frac{\log n}{n}\\
&\leq &\frac{10}{\log b}\frac{1}{\log n}+\frac{e}{\log^2b}\frac{16}{\log n}\ \ \text{(as $16n\geq \log^2 n$)}\\
&= &\left(\frac{10}{\log b}+\frac{16e}{\log^2b}\right)\frac{1}{\log n}
\end{eqnarray*}
This completes the proof.
\end{proof}

\section{Proofs of Two Probability Lemmas}\label{sec-distlem}

\subsection{Proof of Lemma~\ref{lem-stat}}

Let $f:\mathbb{Z}_m \rightarrow \mathbb{C}$ be any function. Recall that, for $0 \leq j \leq m-1$, the Fourier coefficients of $f$ are given by \[\hat{f}\left(j\right)=\frac{1}{m}\sum_{x \in \Z_m} f\left(x\right)\exp\left(-2\pi i jx/m\right).\]
 It is well known that the set of functions $\{\exp\left(2\pi i jx/m\right)\}_{0 \leq j \leq m-1}$ is an orthonormal basis for all functions of the above form, and that $f$ can be expressed as \[f\left(x\right)=\sum_{j=0}^{m-1}\hat{f}\left(j\right)\exp\left(2\pi i jx/m\right).\]
Let us consider $f:\Z_m \rightarrow [0,1]$. Thus, Parseval's identity states that \[\sum_{j=0}^{m-1}\left|\hat{f}\left(j\right)\right|^2 =\frac{1}{m}\sum_{x \in \Z_m}f\left(x\right)^2  \leq 1.\]
 Observe that as $\U_m\left(x\right)=1/m$ is the constant function, $\hat{\U_m}\left(j\right)=0$ for $j \neq 0$. Also, for any distribution $\mu$, $\hat{\mu}\left(0\right)=1/m$. Now
\begin{eqnarray*}
2\epsilon&\leq & \sum_{x \in \Z_m}\left|\mu\left(x\right)-\U_m\left(x\right)\right|\\
&\leq & \sqrt{m}\sqrt{\sum_{x \in \Z_m}\left|\mu\left(x\right)-\U_m\left(x\right)\right|^2} \qquad \text{(Cauchy Schwartz Inequality)}\\
&= & m\sqrt{\sum_{i=0}^{m-1}\left|\left(\hat{\mu}\left(i\right)-\hat{\U_m}\left(i\right)\right)\right|^2}  \\
&= & m\sqrt{\sum_{i=1}^{m-1}\left|\hat{\mu}\left(i\right)\right|^2}  \qquad \text{($\hat{\U_m}\left(j\right)=0$ for $j \neq 0$, and $\hat{\mu}\left(0\right)=\hat{\U_m}\left(0\right)=1/m$)}\\
&\leq & m^{3/2}\max_{i \neq 0}\left\{\left|\hat{\mu}\left(i\right)\right|\right\}.
\end{eqnarray*}
Thus, for some $j \neq 0$, we have \[\left|\hat{\mu}\left(j\right)\right| \geq \frac{2\epsilon}{m^{3/2}}.\]
Observe that
\begin{eqnarray*}
\hat{\mu}\left(j\right)&=&\frac{1}{m}\sum_{x \in \Z_m} \mu\left(x\right)\exp\left(-2\pi i jx/m\right) \\
&=&\frac{1}{m}\E_{x \sim \mu}\left[\exp\left(-2\pi i jx/m\right)\right]\\
&=&\frac{1}{m}\E_{x \sim \mu}\left[\left(\omega^{m-j}\right)^{x}\right].
\end{eqnarray*}
Let $j'=m-j$. Thus, $\left|\hat{\mu}\left(j\right)\right| \geq \frac{2\epsilon}{m^{3/2}}$ implies that \[\left|\E_{x \sim \mu}\left[\left(\omega^{j'}\right)^{x}\right]\right| \geq \frac{2\epsilon}{\sqrt{m}}.\]
This concludes the proof. \qed

\subsection{Proof of Lemma \ref{lem-cp}}

\begin{eqnarray*}
\epsilon^2  &\leq & \left|\left(\sum_{w_1,w_2\in \Z_m^n} \mu_1\left(w_1\right)\mu_2\left(w_2\right)\left[\omega^{\langle w_1,w_2 \rangle}\right]\right)\right|^2\\
&\leq & \left(\sum_{w_1\in \Z_m^n} \mu_1\left(w_1\right)\left|\sum_{w_2 \in \Z_m^n}\mu_2\left(w_2\right)\left[\omega^{\langle w_1,w_2 \rangle}\right]\right|\right)^2\\
&\leq&\left(\sum_{w_1\in \Z_m^n}\mu_1\left(w_1\right)^2\right)\left(\sum_{w_1\in \Z_m^n} \left|\left(\sum_{w_2\in \Z_m^n} \mu_2\left(w_2\right)\left[\omega^{\langle w_1,w_2 \rangle}\right]\right)\right|^2\right)\\
&=&\left(\cp\left(\mu_1\right)\right)\left(\sum_{w_1 \in \Z_m^n} \sum_{w_2',w_2'' \in \Z_m^n} \mu_2\left(w_2'\right)\mu_2\left(w_2''\right)\left[\omega^{\langle w_1,w_2'-w_2'' \rangle}\right]\right)\\
&=&\left(\cp\left(\mu_1\right)\right)\left(\sum_{w \in \Z_m^n} \mu_2\left(w\right)^2m^n\right)\\
&=&m^n \cdot \cp\left(\mu_1\right)\cp\left(\mu_2\right)
\end{eqnarray*}
\qed

\end{document}